\documentclass{article}

\usepackage[nature]{arxiv}

\usepackage{amsthm}
\theoremstyle{plain}
\newtheorem{proposition}{Proposition}[section]
\newtheorem{corollary}[proposition]{Corollary}
\theoremstyle{remark}
\newtheorem{remark}[proposition]{Remark}

\title{The KV Cache Is the New Memory Wall}

\author{%
  \authororcid{Tejinder~Singh}{0000-0002-5870-6204} \\[1em]
  Dell Technologies, Santa Clara, CA 95054, USA \\[0.5em]
  \texttt{Singh.Tejinder@Dell.com}
}

\date{}

\renewcommand{\headeright}{}
\renewcommand{\undertitle}{}
\renewcommand{\shorttitle}{The KV Cache Is the New Memory Wall}

\hypersetup{
  pdftitle={The KV Cache Is the New Memory Wall},
  pdfsubject={Systems for Machine Learning, LLM Inference, Memory Hierarchies},
  pdfauthor={Tejinder Singh},
  pdfkeywords={KV-cache, roofline model, LLM inference, memory hierarchy, memory wall}
}

\newif\ifshowtoc
\showtocfalse
\showtoctrue

\begin{document}
\maketitle

\begin{abstract}
Autoregressive Large Language Model inference at long context is bounded by memory bandwidth rather than arithmetic throughput, and the binding resource shifts from model weights to the Key-Value (KV) cache as sequence length grows. For Llama-3-70B in BF16, the 140\,GB weight footprint alone exceeds the 80\,GB HBM of a single accelerator, and the KV cache of one 128k-token sequence adds 42\,GB on top. Techniques that compress, evict, page, share, or offload KV state have proliferated, yet reported gains are measured under mutually inconsistent workloads, hardware targets, and quality metrics, which prevents rigorous cross-paper comparison. This Systematization of Knowledge paper unifies the field on analytical ground, with a measurement protocol that keeps derived and reported claims strictly separated. We formalize the exact memory traffic of decode-phase attention and derive closed-form expressions for arithmetic intensity as a decaying function of context length, parameterized by hardware topology for NVIDIA H100, NVIDIA B200, and AMD MI300X, including per-die bandwidth partitioning on multi-die packages and the crossover context lengths at which KV traffic overtakes weight traffic. We classify the literature into five domains: quantization, token eviction, KV paging, prefix caching, and heterogeneous memory tiering. We then evaluate a representative method from each domain on common workloads at 128k context under one measurement protocol, reporting compression factor, memory footprint, derived throughput-speedup bounds, and as-reported quality deltas, with the two evidence classes never mixed. The central finding is a three-regime structure: below a hardware-specific crossover context length, weight traffic dominates and KV compression yields negligible end-to-end speedup. Beyond the crossover, KV traffic dominates and each domain trades a measurable quality delta for bandwidth savings that approach the roofline bound. Paging and prefix sharing are lossless with respect to model output but address capacity rather than bandwidth. Quantization and eviction reduce bandwidth directly, with quality degradation that accelerates below 4-bit precision for naive schemes and that turns discontinuous for eviction on position-sensitive tasks, while structure-aware methods extend both floors. Heterogeneous tiering converts the bandwidth wall into an interconnect problem bounded by PCIe or NVLink fabric throughput rather than HBM. We conclude with design rules for selecting a compression domain given a hardware target, context length, and quality budget, and we enumerate open problems in the co-design of compression formats with attention kernels and serving schedulers.
\end{abstract}


% ------------------------------------------------------------------
% Table of contents
% Toggle comment ruled or no-ruled to toggle ToC style, keep one only
% ------------------------------------------------------------------
\ifshowtoc
\ruledtableofcontents
%\tableofcontents
\clearpage
\fi

\section{Introduction and the Decode-Phase Memory Wall}
\label{sec:introduction}

Autoregressive Large Language Models generate text one token at a time. Every decode step executes a full forward pass over the model and attends over the Key-Value (KV) cache, the tensor pair that stores the key and value projections of every token processed so far \cite{vaswani2017attention}. The arithmetic structure of this step is fixed by the transformer architecture, and its cost structure is fixed by the memory hierarchy of the accelerator that executes it. Test-time reasoning workloads, which generate extended chains of intermediate tokens before producing a final answer, push per-request context lengths from $10^{3}$ toward $10^{5}$ to $10^{6}$ tokens \cite{llama3ai2024llama3}. Under these conditions the cost of decode is the cost of moving bytes from High Bandwidth Memory (HBM) to the compute units, not the cost of arithmetic, a bound that recurs at every layer of the deployment stack \cite{singh2026deployment}. This paper systematizes the techniques that attack that bottleneck.

\subsection{Decode-Phase Arithmetic Intensity}

The decode phase is memory bound by construction. At batch size one, a forward pass over a model with $P$ parameters stored at width $w_{p}$ bytes performs approximately $2P$ floating point operations and moves approximately $P \cdot w_{p}$ weight bytes, giving an arithmetic intensity of $2/w_{p}$ FLOP per byte, or $1$ FLOP/B for BF16 weights. The ridge point of an NVIDIA H100 SXM, with $989.5$ TFLOP/s of dense BF16 throughput and $3.35$ TB/s of HBM3 bandwidth, sits at $I^{*} \approx 295$ FLOP/B \cite{williams2009roofline, nvidia2023h100}. Decode at batch size one therefore operates at $1/295$ of the ridge, and no kernel-level improvement can raise sustained compute utilization above a fraction of one percent.

Attention over the KV cache sits in the same regime. One query token attending over $s$ cached tokens costs $4 s d_{h}$ FLOPs per head, excluding softmax terms of lower order in $d_{h}$, and reads $2 s d_{h} w_{kv}$ bytes per head, an intensity of $2/w_{kv}$ FLOP/B, again $1$ FLOP/B in BF16. Grouped-Query Attention raises this figure by the group factor $g = n_{q}/n_{kv}$ because one stored key-value pair serves $g$ query heads, yet even at $g = 8$ the resulting $8$ FLOP/B remains more than an order of magnitude below the H100 ridge \cite{ainslie2023gqa}.

Batching does not rescue the long-context regime. Weight reads are shared across the batch, so weight-side intensity grows linearly in batch size $b$ and reaches the ridge only near $b \approx 295$ on H100 with BF16 weights. KV reads are not shared. Every sequence owns its cache, so KV bytes scale with $b$ exactly as attention FLOPs do, and the attention-side intensity is invariant to batch size. As context length grows, KV traffic overtakes weight traffic, and the batch-size lever that serving systems rely on for throughput loses its effect. Section~3 derives this crossover exactly.

\subsection{KV-Cache Growth with Context Length}

The KV cache grows linearly in context length. For a model with $L$ layers, $n_{kv}$ key-value heads per layer, head dimension $d_{h}$, and element width $w_{kv}$ bytes, the cache of one sequence of length $s$ occupies

\begin{equation}
B_{kv}(s) = 2 \, L \, n_{kv} \, d_{h} \, s \, w_{kv},
\label{eq:kv_bytes_intro}
\end{equation}

where the factor of two accounts for keys and values. For Llama-3-70B ($L = 80$, $n_{kv} = 8$, $d_{h} = 128$) in BF16, Eq.~\eqref{eq:kv_bytes_intro} yields $327{,}680$ bytes per token. The $140$ GB BF16 weight footprint already exceeds the $80$ GB HBM of a single accelerator, and one 128k-token sequence adds $42$ GB of KV state on top \cite{llama3ai2024llama3, nvidia2023h100}. At $10^{6}$ tokens the same sequence requires $328$ GB, the aggregate HBM of more than four H100 devices. Architectural variants compress the constant but not the scaling law. Multi-Query Attention and Grouped-Query Attention reduce $n_{kv}$ by sharing key-value heads across query heads \cite{shazeer2019mqa, ainslie2023gqa}, and Multi-head Latent Attention projects keys and values into a low-rank latent representation \cite{deepseekai2024deepseekv2}. All three leave $B_{kv}$ linear in $s$. Figure~\ref{fig:kv_scaling} plots Eq.~\eqref{eq:kv_bytes_intro} for the Llama-3 family against single-device HBM capacities. The curves cross device capacity at context lengths that production reasoning workloads already exceed.

% Figure 3. KV cache footprint versus context length, Llama-3 family, BF16.
% Per-token KV bytes from Eq. kv_bytes_intro. 8B = 131072, 70B = 327680, 405B = 516096 B/token.
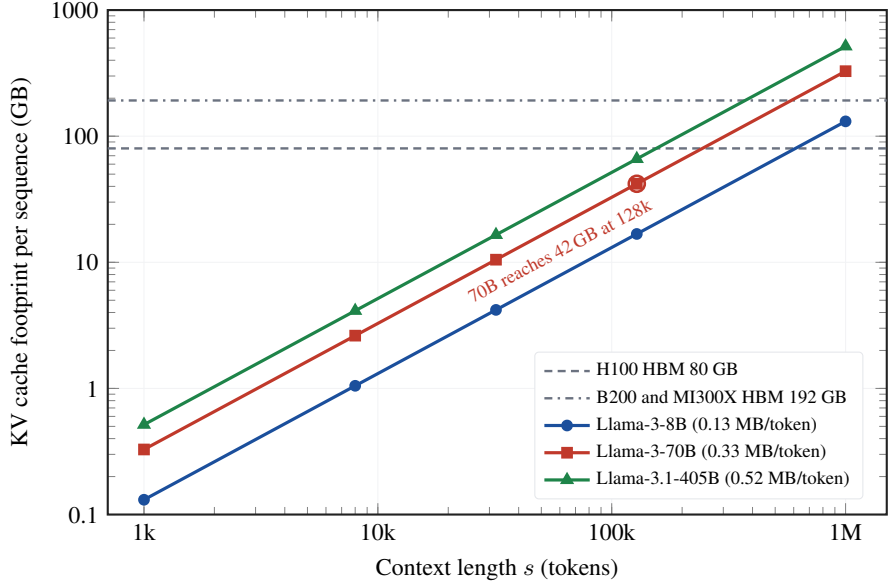
\begin{figure}[t]
\centering
\begin{tikzpicture}
\begin{loglogaxis}[
  width=0.72\linewidth,
  height=0.5\linewidth,
  xlabel={Context length $s$ (tokens)},
  ylabel={KV cache footprint per sequence (GB)},
  xmin=7e2, xmax=1.5e6,
  ymin=0.1, ymax=1000,
  xtick={1e3,1e4,1e5,1e6},
  xticklabels={1k, 10k, 100k, 1M},
  ytick={0.1,1,10,100,1000},
  yticklabels={0.1, 1, 10, 100, 1000},
  grid=major,
  grid style={draw=papergrid!50, line width=0.3pt},
  axis line style={draw=paperaxis, thick},
  tick label style={font=\small},
  label style={font=\small},
  legend style={font=\scriptsize, draw=papergrid, fill=paperwhite, rounded corners=1.5pt, cells={anchor=west}},
  legend pos=south east,
  clip=false
]

% HBM capacity reference lines, drawn first so model curves sit on top
\addplot[domain=7e2:1.5e6, samples=2, color=papermuted, line width=0.9pt, densely dashed] {80};
\addplot[domain=7e2:1.5e6, samples=2, color=papermuted, line width=0.9pt, dashdotted] {192};

% Llama-3-8B. 0.13 MB/token
\addplot[color=paperprimary, line width=1.2pt, mark=*, mark size=1.6pt] coordinates {
  (1e3,0.131) (8e3,1.049) (3.2e4,4.194) (1.28e5,16.78) (1e6,131.07)
};
% Llama-3-70B. 0.33 MB/token
\addplot[color=papersecondary, line width=1.2pt, mark=square*, mark size=1.6pt] coordinates {
  (1e3,0.328) (8e3,2.621) (3.2e4,10.49) (1.28e5,41.94) (1e6,327.68)
};
% Llama-3.1-405B. 0.52 MB/token
\addplot[color=papertertiary, line width=1.2pt, mark=triangle*, mark size=1.9pt] coordinates {
  (1e3,0.516) (8e3,4.129) (3.2e4,16.51) (1.28e5,66.06) (1e6,516.10)
};

\legend{H100 HBM 80 GB, B200 and MI300X HBM 192 GB, Llama-3-8B (0.13 MB/token), Llama-3-70B (0.33 MB/token), Llama-3.1-405B (0.52 MB/token)}

% Crossing annotation for the 70B model at 128k: ring marker on the point,
% label slope-aligned in the empty channel between the 70B and 8B curves.
\addplot[only marks, mark=o, mark size=3pt, line width=1pt, color=papersecondary, forget plot] coordinates {
  (1.28e5,41.94)
};
\node[font=\scriptsize, text=papersecondary, rotate=27.5]
  at (axis cs:6e4, 12.4) {70B reaches 42\,GB at 128k};

\end{loglogaxis}
\end{tikzpicture}
\caption{KV cache footprint per sequence versus context length for the Llama-3 family in BF16, computed from Eq.~\eqref{eq:kv_bytes_intro} \cite{llama3ai2024llama3}. Horizontal reference lines mark single-device HBM capacity. A Llama-3-70B sequence crosses one 80 GB H100 at 244k tokens of context and a Llama-3.1-405B sequence crosses it at 155k. At $10^{6}$ tokens the 70B and 405B caches reach 328 GB and 516 GB, exceeding the aggregate HBM of four and six H100 devices respectively.}
\label{fig:kv_scaling}
\end{figure}

\subsection{Five Fragmented Mitigation Families}

Five families of techniques attack the KV-cache wall, each through a distinct mechanism and each evaluated largely in isolation. \emph{Quantization} stores keys and values at reduced precision, from 8-bit down to 2-bit per element \cite{liu2024kivi, hooper2024kvquant}. \emph{Token eviction} discards cached tokens predicted to carry little future attention weight, retaining sinks and heavy hitters \cite{xiao2023streamingllm, zhang2023h2o, li2024snapkv}. \emph{KV paging} manages the cache in fixed-size blocks to eliminate fragmentation and enable fine-grained sharing across sequences \cite{kwon2023vllm}. \emph{Prefix caching} reuses the KV state of shared prompt prefixes across requests \cite{zheng2024sglang}. \emph{Heterogeneous memory tiering} migrates KV state across HBM, host DRAM, and NVMe storage, substituting interconnect bandwidth for HBM residency \cite{sheng2023flexgen, liu2024cachegen, lmcache2024}. The mechanisms are complementary in principle, yet the literature evaluates them under mutually inconsistent context lengths, batch sizes, hardware targets, and quality metrics, so reported gains cannot be composed and can barely be compared. Xu et al.~catalog this space descriptively \cite{xu2026kvstrategies}; our aim is not another catalog but the analytical frame that makes the catalog comparable.

\subsection{Comparability Deficiencies in the Published Evidence Base}

Three methodological deficiencies recur across the five families. First, there is no shared analytical model of KV memory traffic, so papers report speedups against baselines whose byte counts differ through batching policy, kernel choice, and cache layout. Second, hardware topology is ignored. NVIDIA B200 packages two compute dies behind an NV-HBI die-to-die link, and AMD MI300X partitions compute across eight accelerator chiplets, so aggregate HBM bandwidth is not uniformly accessible to every streaming multiprocessor \cite{nvidia2024b200, amd2024mi300x}. A technique that reduces KV bytes by a fixed factor can yield markedly different end-to-end speedups depending on how cache blocks map onto dies and chiplets. Third, quality is reported on incompatible scales. Perplexity on one corpus, needle retrieval accuracy on another, and downstream task scores on a third, with almost no study reporting quality and throughput under one protocol. Almost every published number is correct in isolation. Almost none are jointly interpretable.

\subsection{Systematization Goal and Contributions}

The goal of this paper is to place the KV-cache mitigation literature on analytical and empirical ground that permits rigorous comparison. We make four contributions.

\begin{enumerate}[leftmargin=1.5em]
  \item We formalize the KV-cache structure of autoregressive transformers and classify the mitigation literature into five domains, stating the optimization problem each domain solves and the mechanism each representative method applies (Section~2).
  \item We derive closed-form expressions for decode arithmetic intensity as a decaying function of context length, with exact per-kernel FLOP and byte accounting (Section~3).
  \item We instantiate the roofline model with hardware topology constraints for NVIDIA H100, NVIDIA B200, and AMD MI300X, modeling per-die bandwidth partitioning on multi-die packages and deriving the crossover context lengths at which KV traffic overtakes weight traffic (Section~3).
  \item We evaluate one representative method per domain on standardized 128k-context workloads under a single protocol that separates derived bounds from as-reported quality deltas (Section~4).
\end{enumerate}

Section~5 analyzes cross-cutting challenges, including compression-aware kernel design, scheduler interaction, and evaluation gaps, and Section~6 concludes with design rules for selecting a compression domain given a hardware target, a context length, and a quality budget.
\section{KV Cache Anatomy and a Five-Domain Taxonomy of Compression}
\label{sec:taxonomy}

This section fixes notation for the KV cache of an autoregressive transformer, quantifies the per-token byte cost across attention variants, and organizes the mitigation literature into five domains distinguished by which factor of a unified footprint model each one scales.

\subsection{Exact KV-Cache Structure of Autoregressive Transformers}

Consider a decoder-only transformer with $L$ layers, hidden size $d$, $n_{q}$ query heads and $n_{kv}$ key-value heads per layer, and head dimension $d_{h} = d / n_{q}$. At decode step $t$ the network computes, for every layer $\ell$, the projections $q_{t} = x_{t} W_{q}$, $k_{t} = x_{t} W_{k}$, and $v_{t} = x_{t} W_{v}$, and appends $k_{t}$ and $v_{t}$ to the cache tensors
\begin{equation}
K_{\ell}, V_{\ell} \in \mathbb{R}^{\,n_{kv} \times t \times d_{h}}, \qquad \ell = 1, \dots, L.
\label{eq:kv_tensors}
\end{equation}
The attention output for the new token is
\begin{equation}
o_{t} = \mathrm{softmax}\!\left(\frac{q_{t} K_{\ell}^{\top}}{\sqrt{d_{h}}}\right) V_{\ell},
\label{eq:decode_attention}
\end{equation}
evaluated per query head against the shared key-value slices. Two asymmetries define the serving problem. Prefill writes all $s_{0}$ prompt positions in parallel and is compute bound, with intensity near $2 s_{0} / w$ FLOP/B for short, weight-dominated prompts and higher still at long context, where attention FLOPs overtake weight FLOPs. Decode appends one position and reads the entire cache of $t$ positions, with intensity $2 g / w_{kv}$ per layer as derived in Section~1. The cache is therefore write-once per token and read-forever, and its footprint follows Eq.~\eqref{eq:kv_bytes_intro}.

Table~\ref{tab:model_kv} evaluates Eq.~\eqref{eq:kv_bytes_intro} for six production architectures. Three observations follow. First, GQA cuts the per-token cost by the group factor $g$, which is why Llama-3-70B ($g = 8$) stores only $0.33$ MB per token while the smaller MHA-based Llama-2-13B stores $0.82$ MB. Second, the constant grows with depth, so the 405B model pays $0.52$ MB per token despite the same group factor. Third, Multi-head Latent Attention stores a low-rank latent vector of $d_{c} = 512$ dimensions plus $64$ rotary dimensions per layer rather than full key-value pairs, yielding $69{,}120$ bytes per token for DeepSeek-V2, roughly $5\times$ below the GQA designs \cite{deepseekai2024deepseekv2}. Architectural choice shifts the constant by up to $12\times$ across the table. It does not change the linear growth in $s$.

\begin{table}[t]
\centering
\caption{Per-token KV-cache footprint of representative architectures in BF16, computed from Eq.~\eqref{eq:kv_bytes_intro}. $^{\dagger}$MLA stores a compressed latent vector of $(d_{c} + d_{r}) = 576$ dimensions per layer instead of full key-value pairs.}
\label{tab:model_kv}
\footnotesize
\begin{tabular}{llrrrrr}
\toprule
Model & Attn. & $L$ & $n_{q}$ & $n_{kv}$ & $d_{h}$ & MB/token \\
\midrule
Llama-2-13B \cite{touvron2023llama2} & MHA & 40 & 40 & 40 & 128 & 0.82 MB \\
Mixtral 8x7B \cite{jiang2024mixtral} & GQA & 32 & 32 & 8 & 128 & 0.13 MB \\
Llama-3-8B \cite{llama3ai2024llama3} & GQA & 32 & 32 & 8 & 128 & 0.13 MB \\
Llama-3-70B \cite{llama3ai2024llama3} & GQA & 80 & 64 & 8 & 128 & 0.33 MB \\
Llama-3.1-405B \cite{llama3ai2024llama3} & GQA & 126 & 128 & 8 & 128 & 0.52 MB \\
DeepSeek-V2 \cite{deepseekai2024deepseekv2} & MLA & 60 & 128 & 576$^{\dagger}$ & --- & 0.07 MB \\
\bottomrule
\end{tabular}
\end{table}

\subsection{A Unified Footprint Model}

Every mitigation technique in the literature reduces to scaling one factor of the cache footprint or relocating the bytes across the memory hierarchy. We make this explicit with a unified model,
\begin{equation}
B_{kv}^{\mathrm{eff}}(s) = \underbrace{2 \, L \, n_{kv} \, d_{h}}_{\text{architecture}} \cdot \underbrace{w_{q}}_{\text{quantization}} \cdot \underbrace{r}_{\text{eviction}} \cdot \underbrace{s}_{\text{context}} \cdot \underbrace{\sigma^{-1}}_{\text{sharing}} \cdot \underbrace{(1 + \phi)}_{\text{fragmentation}}.
\label{eq:kv_bytes_unified}
\end{equation}
Here $w_{q}$ is the stored element width after quantization, $r \in (0, 1]$ the retention fraction under eviction, $\sigma \ge 1$ the cross-request sharing factor from prefix caching, and $\phi \ge 0$ the fractional waste from allocation granularity. Memory tiering adds a residency vector $\lambda = (\lambda_{\mathrm{hbm}}, \lambda_{\mathrm{dram}}, \lambda_{\mathrm{nvme}})$ that records which fraction of the cache occupies each tier. Tiering does not shrink the byte count. It determines which interconnect carries the bytes on every read.

The model is the spine of the taxonomy. Quantization sets $w_{q}$. Eviction sets $r$. Paging drives $\phi$ toward zero. Prefix caching sets $\sigma$. Tiering sets $\lambda$. The domains are orthogonal by construction at the level of byte accounting and compose multiplicatively in footprint. Quality and kernel interactions do not factor as cleanly, and Sections~4 and~5 analyze the coupling. Figure~\ref{fig:taxonomy} presents the taxonomy with representative methods per domain.

% Figure 2. Five-domain taxonomy of KV-cache mitigation techniques.
% Each domain is annotated with the factor it scales in Eq. kv_bytes_unified.
\begin{figure}[t]
\centering
\begin{tikzpicture}[
  font=\small,
  dom/.style={rectangle, rounded corners=3pt, line width=0.9pt, align=center, inner sep=5pt, text width=2.7cm, font=\scriptsize\bfseries},
  det/.style={rectangle, rounded corners=3pt, draw=papergrid!80!paperaxis, fill=paperbg, align=left, inner sep=5pt, text width=8.9cm, font=\scriptsize}
]
\node[paper secondary node, font=\small\bfseries] (root) at (0,0) {KV-cache\\mitigation};

\node[dom, draw=paperprimary!85!paperblack, fill=paperprimary!12] (d1) at (3.4,3.0) {Quantization\\[1pt] scales $w_{q}$};
\node[dom, draw=papersecondary!85!paperblack, fill=papersecondary!12] (d2) at (3.4,1.5) {Token eviction\\[1pt] scales $r$};
\node[dom, draw=papertertiary!85!paperblack, fill=papertertiary!12] (d3) at (3.4,0) {KV paging\\[1pt] scales $\phi$};
\node[dom, draw=paperaccent!85!paperblack, fill=paperaccent!12] (d4) at (3.4,-1.5) {Prefix caching\\[1pt] scales $\sigma$};
\node[dom, draw=paperquaternary!85!paperblack, fill=paperquaternary!12] (d5) at (3.4,-3.0) {Memory tiering\\[1pt] sets residency $\lambda$};

\node[det, anchor=west] (m1) at (5.6,3.0) {KIVI. 2-bit, per-channel keys, per-token values \cite{liu2024kivi}. KVQuant. Sub-4-bit with outlier isolation \cite{hooper2024kvquant}. QServe W4A8KV4 co-design \cite{lin2024qserve}.};
\node[det, anchor=west] (m2) at (5.6,1.5) {StreamingLLM. Sinks plus sliding window \cite{xiao2023streamingllm}. H2O. Heavy-hitter retention \cite{zhang2023h2o}. SnapKV. Observation-window selection \cite{li2024snapkv}. Keyformer \cite{adnan2024keyformer}. PyramidKV. Per-layer retention budgets \cite{cai2024pyramidkv}. Read-path analog: Quest \cite{tang2024quest}, MInference \cite{jiang2024minference}.};
\node[det, anchor=west] (m3) at (5.6,0) {PagedAttention. Fixed-size blocks, block tables, copy-on-write sharing \cite{kwon2023vllm}. Internal waste bounded to the final block.};
\node[det, anchor=west] (m4) at (5.6,-1.5) {RadixAttention. Radix-tree reuse with cache-aware scheduling \cite{zheng2024sglang}. LMCache cross-request deduplication \cite{lmcache2024}. CacheBlend. Selective re-encoding for chunk reuse \cite{yao2025cacheblend}.};
\node[det, anchor=west] (m5) at (5.6,-3.0) {FlexGen. Offload to host DRAM and NVMe \cite{sheng2023flexgen}. CacheGen. Compressed KV bitstream transport \cite{liu2024cachegen}. LMCache. Pipelined multi-tier transfer \cite{lmcache2024}.};

% Orthogonal trunk-and-branch connectors: horizontal stub from the root,
% shared vertical trunk, horizontal arrow into each domain box.
\draw[paper arrow] (root.east) -- ++(0.55,0) |- (d1.west);
\draw[paper arrow] (root.east) -- ++(0.55,0) |- (d2.west);
\draw[paper arrow] (root.east) -- ++(0.55,0) |- (d3.west);
\draw[paper arrow] (root.east) -- ++(0.55,0) |- (d4.west);
\draw[paper arrow] (root.east) -- ++(0.55,0) |- (d5.west);
\draw[paper arrow] (d1.east) to (m1.west);
\draw[paper arrow] (d2.east) to (m2.west);
\draw[paper arrow] (d3.east) to (m3.west);
\draw[paper arrow] (d4.east) to (m4.west);
\draw[paper arrow] (d5.east) to (m5.west);
\end{tikzpicture}
\caption{Five-domain taxonomy of KV-cache mitigation techniques. Each domain is annotated with the factor it scales in the unified footprint model of Eq.~\eqref{eq:kv_bytes_unified}. Quantization and eviction trade quality for bytes. Paging and prefix caching are lossless and address capacity and duplication. Tiering relocates bytes across the memory hierarchy and shifts the bottleneck from HBM to interconnect bandwidth.}
\label{fig:taxonomy}
\end{figure}
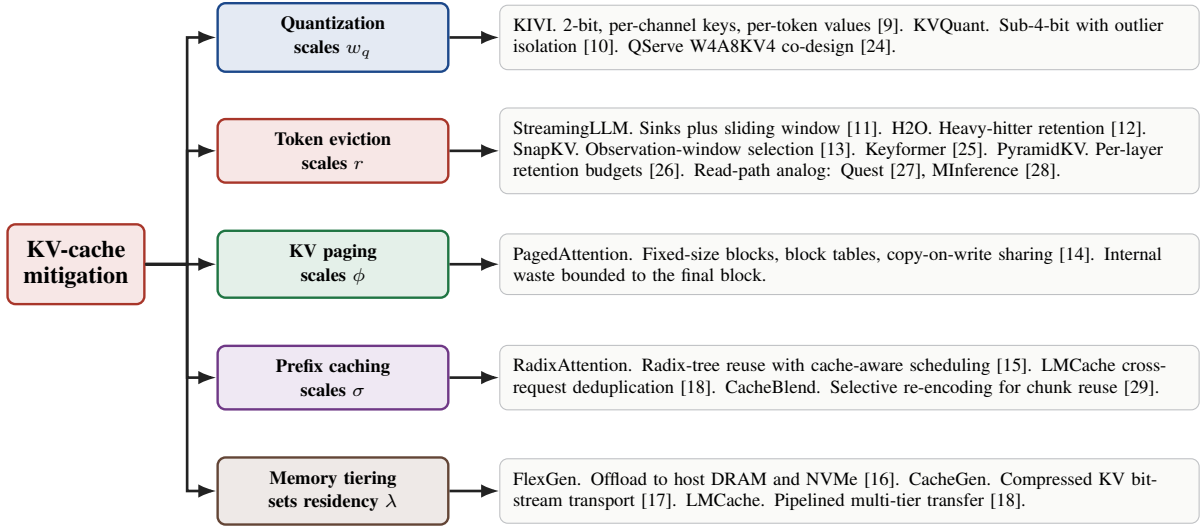

\subsection{Domain I. Precision Reduction of Stored Keys and Values}

Quantization sets $w_{q} = w / c_{q}$ for compression factor $c_{q} \in \{2, 4, 8\}$, mapping 16-bit elements to 8, 4, or 2 bits. The mechanism is not free. Keys exhibit large-magnitude outlier channels that dominate the quantization error, while values are better behaved, which motivates asymmetric layouts. KIVI stores keys per-channel and values per-token at 2 bits, holding a small window of recent tokens at full precision \cite{liu2024kivi}. KVQuant combines per-channel key quantization, pre-RoPE key storage, and dense-and-sparse outlier isolation, reaching 3-bit with sub-0.1 perplexity degradation and 2-bit with modest task-dependent loss \cite{hooper2024kvquant}. QServe demonstrates that the format must be co-designed with the kernels, choosing W4A8KV4 so dequantization stays off the critical path \cite{lin2024qserve}. The bandwidth saving is exact and immediate. Every byte removed from the stored element is a byte not read from HBM on every decode step.

\subsection{Domain II. Token Eviction and Retention Policies}

Eviction sets a retention fraction $r < 1$ by keeping a subset $S_{t} \subseteq \{1, \dots, t\}$ of positions and discarding the rest. The decision rule distinguishes methods. StreamingLLM retains the first few attention-sink tokens plus a sliding window, enabling multi-million-token streams at fixed memory but forfeiting retrieval of evicted middle content \cite{xiao2023streamingllm}. H2O retains heavy hitters selected by accumulated attention scores \cite{zhang2023h2o}. SnapKV selects positions per head from an observation window at the prompt tail, closely matching long-context task accuracy at retention near 25 percent and below on aggregate benchmarks \cite{li2024snapkv}. Keyformer derives discard scores from logit perturbations \cite{adnan2024keyformer}, and FastGen assigns per-head eviction policies \cite{ge2023fastgen}. Two refinements sharpen the retention mechanism. PyramidKV replaces the uniform budget with a per-layer schedule, retaining near-full caches in the lower layers, where attention is diffuse, and shrinking budgets toward the upper layers, where attention concentrates \cite{cai2024pyramidkv}. A parallel family keeps the full cache resident and moves selection from the write path to the read path. Quest selects a query-dependent subset of KV pages at every decode step \cite{tang2024quest}, and MInference classifies heads into dynamic sparse patterns that read only a fraction of the cache during prefill and decode \cite{jiang2024minference}. In the unified model, read-path selection leaves every factor of Eq.~\eqref{eq:kv_bytes_unified} unchanged and instead discounts the per-step traffic of Eq.~\eqref{eq:decode_bytes} by a query-dependent fraction. It buys bandwidth reversibly but offers no capacity relief, so it complements eviction rather than replacing it. The risk is structural. Eviction is irreversible, and tasks whose answers depend on arbitrary past positions degrade discontinuously once $r$ falls below the information density of the prompt.

\subsection{Domain III. Paged Block Management of KV Memory}

Paging scales no per-token factor at all. It eliminates the fragmentation term $\phi$ by allocating the cache in fixed-size blocks of $\tau$ tokens, typically $\tau = 16$, tracked through a block table in the style of virtual memory \cite{kwon2023vllm}. Reservation-based allocation strands 60 to 80 percent of cache memory in over-reservation and fragmentation on production traces, while paging bounds waste to the partially filled final block of each sequence, measured below 4 percent \cite{kwon2023vllm}. Copy-on-write block sharing makes parallel sampling and beam search nearly free in memory. Paging is lossless by construction and is now the default substrate of production serving stacks. Its limit is that it manages bytes without reducing them.

\subsection{Domain IV. Prefix Sharing Across Requests}

Prefix caching raises the sharing factor $\sigma$ by storing the KV state of a prompt prefix once and reusing it across every request that shares that prefix, a pattern common to system prompts, few-shot exemplars, retrieved document context, and multi-turn history \cite{zheng2024sglang}. RadixAttention organizes cached prefixes in a radix tree with LRU eviction and schedules requests to maximize hit rate, reporting multi-fold throughput gains on sharing-heavy workloads \cite{zheng2024sglang}. LMCache generalizes reuse to non-prefix spans, deduplicating KV state across requests and across serving engines \cite{lmcache2024}. CacheBlend targets retrieval workloads, where the reusable state is a set of independently encoded document chunks rather than one contiguous prefix, and restores cross-chunk attention by re-encoding a small selected fraction of tokens, recovering most of the time-to-first-token benefit of full reuse at a fraction of the recompute cost \cite{yao2025cacheblend}. Exact reuse is lossless and additionally saves prefill compute, since shared positions are never re-encoded; selective-recompute variants are near-lossless up to the re-encoded fraction. Its gains are workload dependent. A request stream with no shared structure sees $\sigma = 1$.

\subsection{Domain V. Heterogeneous Tiering of KV State}

Tiering relocates bytes rather than reducing them. The residency vector $\lambda$ records what fraction of the cache occupies each tier, and the per-step cost follows the tier that must be read. FlexGen offloads weights, activations, and KV state to host memory and disk to maximize throughput on memory-limited GPUs \cite{sheng2023flexgen}. CacheGen compresses KV state into compact bitstreams for network transport, cutting transfer size by roughly $4\times$ \cite{liu2024cachegen}. LMCache pipelines KV movement across GPU, CPU, and disk with deduplication, reporting multi-fold time-to-first-token reductions on multi-turn and retrieval workloads \cite{lmcache2024}. The binding constraint is interconnect bandwidth. PCIe Gen5 x16 delivers roughly 64 GB/s, about $50\times$ below H100 HBM bandwidth, and even NVLink at 900 GB/s sits $3.7\times$ below it \cite{nvidia2023h100}. Tiering therefore pays when reuse distance is long, as in prefix reuse over minutes, and starves the decode loop when per-step reads miss HBM.

\subsection{Orthogonality and Composition of the Five Domains}

Table~\ref{tab:taxonomy} summarizes the taxonomy. Because the domains scale orthogonal factors of Eq.~\eqref{eq:kv_bytes_unified}, they compose multiplicatively in bytes; the composed quality cost is not multiplicative and must be validated jointly. A 2-bit quantization ($w_{q} = w/8$) combined with 50 percent retention ($r = 0.5$) and a sharing factor of $\sigma = 4$ reduces the effective footprint to $1/64$ of the uncompressed baseline before tiering is even considered. Composition is not automatic. It requires attention kernels that accept quantized and evicted layouts and schedulers that exploit sharing and tiering without stalling decode, interactions that Section~5 examines. The design problem is also asymmetric. Paging and prefix sharing are lossless and should always be active. The optimization problem lives entirely in the choice of $(w_{q}, r, \lambda)$ under a quality budget.

\begin{table}[t]
\centering
\caption{Five-domain taxonomy of KV-cache mitigation. Each domain scales one factor of Eq.~\eqref{eq:kv_bytes_unified} or relocates bytes across tiers. $^{\ast}$Transport codecs such as CacheGen optionally quantize, making the link lossy by choice.}
\label{tab:taxonomy}
\footnotesize
\begin{tabularx}{\linewidth}{>{\raggedright\arraybackslash}p{1.9cm}>{\raggedright\arraybackslash}p{2.6cm}cc>{\raggedright\arraybackslash}X}
\toprule
Domain & Mechanism & Scales & Lossless & Representative methods \\
\midrule
Quantization & Store keys and values at $w_{q} < w$ & $w_{q}$ & No & KIVI \cite{liu2024kivi}, KVQuant \cite{hooper2024kvquant}, QServe \cite{lin2024qserve} \\
Token eviction & Retain fraction $r$ of $s$ positions & $r$ & No & StreamingLLM \cite{xiao2023streamingllm}, H2O \cite{zhang2023h2o}, SnapKV \cite{li2024snapkv}, PyramidKV \cite{cai2024pyramidkv} \\
KV paging & Fixed-size block allocation & $\phi$ & Yes & PagedAttention \cite{kwon2023vllm} \\
Prefix caching & Deduplicate shared prefixes & $\sigma$ & Yes & RadixAttention \cite{zheng2024sglang}, LMCache \cite{lmcache2024} \\
Memory tiering & Migrate blocks across HBM, DRAM, NVMe & $\lambda$ & Yes$^{\ast}$ & FlexGen \cite{sheng2023flexgen}, CacheGen \cite{liu2024cachegen} \\
\bottomrule
\end{tabularx}
\end{table}
\section{Analytical Derivations of Context-Dependent Arithmetic Intensity}
\label{sec:roofline}

This section derives the cost of a decode step exactly, expresses arithmetic intensity as a function of context length and batch size, instantiates the roofline model with the topology of three production accelerators, and derives the crossover context lengths that separate the weight-bound and KV-bound regimes. All results are closed-form and require only the model constants of Table~\ref{tab:model_kv} and the hardware constants of Table~\ref{tab:hardware}. Numerical instantiations are derived quantities computed from vendor specifications and model constants, not measurements. Section~4 evaluates representative methods against these bounds.

\subsection{Exact FLOP and Byte Accounting for a Decode Step}

A decode step for a batch of $b$ sequences, each at context length $s$, executes per layer the QKV projections at $2 d d_{h} (n_{q} + 2 n_{kv})$ FLOPs, the attention score and value contractions at $4 s d_{h} n_{q}$ FLOPs, the output projection at $2 d d_{h} n_{q}$ FLOPs, and the MLP at $6 d d_{ff}$ FLOPs for a SwiGLU block with three weight matrices. Folding all weight matrices into $P$ and summing over $L$ layers gives the exact per-step FLOP count
\begin{equation}
F(s, b) = b \left( 2P + 4 L n_{q} d_{h} s \right),
\label{eq:decode_flops}
\end{equation}
where the softmax contributes $3 L n_{q} s b$ additional FLOPs, lower order in $d_{h}$, and is omitted. The corresponding byte traffic, assuming a FlashAttention-class kernel that reads the cache exactly once \cite{dao2022flashattention}, is
\begin{equation}
B(s, b) = P w_{p} + 2 b L n_{kv} d_{h} s w_{kv}.
\label{eq:decode_bytes}
\end{equation}
Weights are read once per step regardless of batch size. KV bytes scale with the batch because each sequence owns its cache. Equations~\eqref{eq:decode_flops} and~\eqref{eq:decode_bytes} are the entire cost model. Everything else in this section is algebra over their ratio.

\subsection{Arithmetic Intensity as a Function of Context Length and Batch}

The decode arithmetic intensity follows directly.
\begin{equation}
\mathrm{AI}(s, b) = \frac{F(s, b)}{B(s, b)} = \frac{b \left( 2P + 4 L n_{q} d_{h} s \right)}{P w_{p} + 2 b L n_{kv} d_{h} s w_{kv}}.
\label{eq:ai_general}
\end{equation}
Its two limits define the operating regimes.
\begin{equation}
\lim_{s \to 0} \mathrm{AI}(s, b) = \frac{2b}{w_{p}}, \qquad \lim_{s \to \infty} \mathrm{AI}(s, b) = \frac{2 n_{q}}{n_{kv} w_{kv}} = \frac{2g}{w_{kv}}.
\label{eq:ai_floor}
\end{equation}
The short-context limit grows with batch. The long-context limit is pinned by the GQA group factor and the KV storage precision, and the batch size cancels exactly. The two limits coincide at the fixed point
\begin{equation}
b^{\dagger} = \frac{g \, w_{p}}{w_{kv}},
\label{eq:fixed_point}
\end{equation}
which for BF16 weights, BF16 cache, and $g = 8$ gives $b^{\dagger} = 8$. Below $b^{\dagger}$, longer context raises intensity toward the floor. Above $b^{\dagger}$, longer context lowers intensity toward the floor. Production serving operates at $b \gg b^{\dagger}$, so in practice decode intensity decays with context length, from the batch-amplified weight-regime value $2b/w_{p}$ toward the KV floor $2g/w_{kv}$. Quantizing the cache raises the floor. At 2-bit storage the floor reaches $64$ FLOP/B, still $4.6\times$ below the H100 ridge. No realistic storage precision makes decode compute bound. Figure~\ref{fig:decode_roofline} places these operating points on the roofline.

% Figure 1. Decode-phase roofline with context-dependent operating points.
% Slopes in TFLOP/s. H100 solid, B200 aggregate dashed, B200 per-die dotted.
\begin{figure}[t]
\centering
\begin{tikzpicture}
\begin{axis}[
  width=0.72\linewidth,
  height=0.5\linewidth,
  xlabel={Arithmetic intensity (FLOP/Byte)},
  ylabel={Attainable performance (TFLOP/s)},
  xmode=log, ymode=log,
  xmin=0.5, xmax=8192,
  ymin=1, ymax=4000,
  grid=major,
  grid style={draw=papergrid!50, line width=0.3pt},
  axis line style={draw=paperaxis, thick},
  tick label style={font=\small},
  label style={font=\small},
  legend style={at={(0.965,0.03)}, anchor=south east, font=\scriptsize, draw=papergrid, fill=paperwhite, rounded corners=1.5pt, cells={anchor=west}, inner sep=4pt},
  extra x ticks={295},
  extra x tick labels={295},
  extra x tick style={grid=none, tick label style={font=\scriptsize, text=paperprimary}},
  clip=false
]
% H100 roofline
\addplot[domain=0.5:8192, samples=400, color=paperprimary, line width=1.2pt] {min(3.35*x, 989.5)};
\addlegendentry{H100. 3.35 TB/s}
% B200 aggregate roofline
\addplot[domain=0.5:8192, samples=400, color=papersecondary, line width=1.0pt, densely dashed] {min(8*x, 2250)};
\addlegendentry{B200 aggregate. 8 TB/s}
% B200 per-die roofline
\addplot[domain=0.5:8192, samples=400, color=papersecondary, line width=0.9pt, loosely dotted] {min(4*x, 1125)};
\addlegendentry{B200 per-die. 4 TB/s}

% ceiling labels (bandwidth labels live in the legend, south east)
\node[font=\scriptsize, text=paperprimary, anchor=north west] at (axis cs:330, 950) {H100 ceiling 989.5 TF};
\node[font=\scriptsize, text=papersecondary, anchor=south west] at (axis cs:330, 2350) {B200 ceiling 2250 TF};

% H100 ridge marker; the ridge value 295 is shown as a blue extra x tick
\draw[thin, densely dotted, draw=paperprimary!60] (axis cs:295, 1) to (axis cs:295, 989.5);
\addplot[only marks, mark=o, mark size=2.4pt, color=paperprimary, forget plot] coordinates {(295, 989.5)};

% operating points, Llama-3-70B on H100, from Eq. ai_general
\addplot[only marks, mark=*, mark size=2.4pt, color=paperprimary, forget plot] coordinates {(1.13, 3.78)};
\node[font=\scriptsize, text=paperprimary, anchor=west] at (axis cs:1.24, 2.9) {decode $b{=}1$, $s{=}8$k};

\addplot[only marks, mark=*, mark size=2.4pt, color=paperprimary, forget plot] coordinates {(2.61, 8.76)};
\node[font=\scriptsize, text=paperprimary, anchor=west] at (axis cs:2.85, 6.8) {decode $b{=}1$, $s{=}128$k};

\addplot[only marks, mark=*, mark size=2.4pt, color=paperprimary, forget plot] coordinates {(50.8, 170.2)};
\draw[draw=paperprimary!55, line width=0.4pt] (axis cs:50.8, 158) to (axis cs:50.8, 68);
\node[font=\scriptsize, text=paperprimary, anchor=north] at (axis cs:50.8, 64) {decode $b{=}256$, $s{=}8$k};

\addplot[only marks, mark=square*, mark size=2.2pt, color=papertertiary, forget plot] coordinates {(8, 26.8)};
\node[font=\scriptsize, text=papertertiary, anchor=south west] at (axis cs:8.6, 20) {BF16 KV floor $2g/w$};

\addplot[only marks, mark=square*, mark size=2.2pt, color=paperaccent, forget plot] coordinates {(64, 214.4)};
\node[font=\scriptsize, text=paperaccent, anchor=west] at (axis cs:72, 200) {2-bit KV floor};

\addplot[only marks, mark=triangle*, mark size=2.6pt, color=papersecondary, forget plot] coordinates {(4096, 692.7)};
\node[font=\scriptsize, text=papersecondary, anchor=north] at (axis cs:4096, 620) {prefill, $s_{0}{=}4$k};

\end{axis}
\end{tikzpicture}
\caption{Roofline with decode operating points for Llama-3-70B on H100, computed from Eq.~\eqref{eq:ai_general} and Table~\ref{tab:hardware}. Solid line. H100 roofline. Dashed and dotted lines. B200 aggregate and per-die rooflines, showing the bandwidth penalty when a sequence is pinned to one die. Marked ceilings are BF16 dense compute rates. Decode points sit between sixfold and more than two-hundredfold below the ridge across batch sizes and context lengths. Only prefill reaches the compute-bound regime. The prefill point uses a 4k-token prompt at 70 percent MFU, matching Figure~\ref{fig:prefix}.}
\label{fig:decode_roofline}
\end{figure}
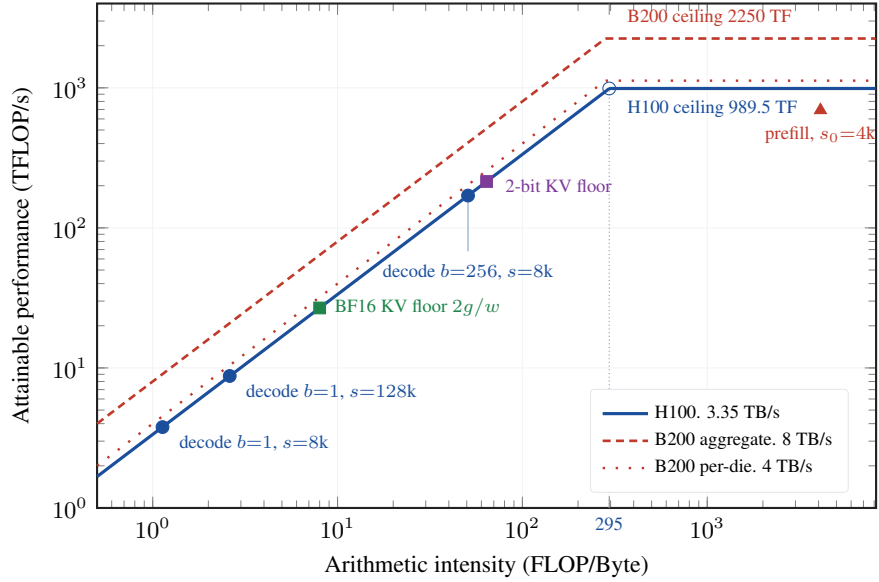

A second consequence follows. The weight-regime intensity reaches the ridge only at $b \approx I^{*} w_{p} / 2 \approx 295$ on H100 with BF16 weights, but at $b = 295$ the crossover context length for Llama-3-70B has already fallen to $1.4$k tokens, as Table~\ref{tab:crossover} extrapolates. The batch size at which weights could become compute bound is precisely the batch size at which any context beyond a single page of tokens is KV bound. Decode is memory bound everywhere that matters.

\subsection{Hardware Topology and Per-Die Bandwidth Partitioning}

The roofline bound $T_{\mathrm{step}} \ge B(s, b)/\beta$ holds with $\beta$ equal to aggregate HBM bandwidth only when every byte is locally accessible to every compute unit at that rate. Table~\ref{tab:hardware} shows that this assumption fails on current multi-die packages.

\begin{table}[t]
\centering
\caption{Hardware specification of the three benchmarked accelerators \cite{nvidia2023h100, nvidia2024b200, amd2024mi300x}. Compute figures are dense tensor core rates in TFLOP/s. Per-partition bandwidth is aggregate bandwidth divided by die or stack count. Ridge point computed against BF16 dense compute. MI300X on-package fabric bandwidth is not publicly disclosed. Its inter-GPU Infinity Fabric runs at 896 GB/s.}
\label{tab:hardware}
\footnotesize
\begin{tabular}{llrrrrrr}
\toprule
Device & Package & HBM & $\beta_{\mathrm{agg}}$ & $\beta_{\mathrm{part}}$ & BF16 & FP8 & $I^{*}$ \\
\midrule
H100 SXM & Monolithic & 80 GB & 3.35 TB/s & 3.35 TB/s & 989.5 & 1979 & 295 \\
B200 & 2 dies, NV-HBI 10 TB/s & 192 GB & 8.0 TB/s & 4.0 TB/s & 2250 & 4500 & 281 \\
MI300X & 8 XCDs, Infinity Fabric & 192 GB & 5.3 TB/s & 0.66 TB/s & 1307 & 2615 & 247 \\
\bottomrule
\end{tabular}
\end{table}

We model a package as $D$ partitions, each with local HBM bandwidth $\beta_{p} = \beta_{\mathrm{agg}}/D$, joined by a die-to-die fabric of bandwidth $\beta_{f}$. A decode step that reads $x$ KV bytes with fraction $\gamma$ resident locally completes the read, under ideal overlap, in
\begin{equation}
T_{\mathrm{read}}(x, \gamma) = \max\!\left(\frac{\gamma x}{\beta_{p}},\; \frac{(1-\gamma) x}{\beta_{f}}\right),
\qquad
\beta_{\mathrm{eff}} = \frac{x}{T_{\mathrm{read}}}.
\label{eq:beta_eff}
\end{equation}
Two placement extremes bound the behavior. Striped placement ($\gamma = 1/D$) on B200 gives $T_{\mathrm{read}} = \max(x/8\,\mathrm{TB/s},\, x/20\,\mathrm{TB/s}) = x/8\,\mathrm{TB/s}$, recovering the full aggregate because NV-HBI at 10 TB/s outruns the per-die rate of 4 TB/s. Pinned placement ($\gamma = 1$) gives $\beta_{\mathrm{eff}} = \beta_{p} = 4$ TB/s, half the aggregate. On MI300X the pinned penalty reaches $8\times$, since one stack slice supplies $0.66$ TB/s. Figure~\ref{fig:topology} depicts the three topologies. The consequence for serving is that KV page placement is a first-order performance decision on multi-die hardware, yet current serving stacks allocate pages without die or stack affinity. We return to this gap in Section~5.

% Figure 10. Package topology and per-partition bandwidth for H100, B200, MI300X.
% Per-die and per-stack splits derived from Table hardware. R_inf ceilings from Eq. tokrate.
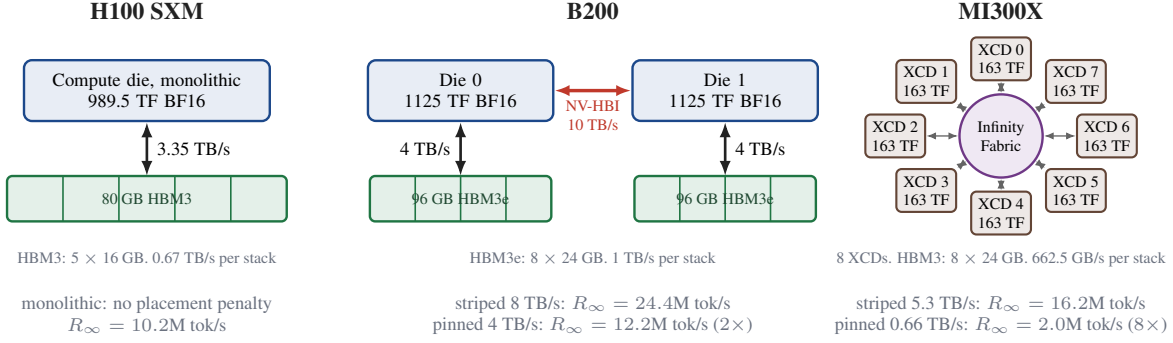
\begin{figure}[t]
\centering
\begin{tikzpicture}[
  font=\scriptsize,
  die/.style={rectangle, rounded corners=3pt, draw=paperprimary!85!paperblack, fill=paperprimary!10, line width=0.9pt, align=center, inner sep=5pt},
  bar/.style={rectangle, rounded corners=2pt, draw=papertertiary!85!paperblack, fill=papertertiary!12, line width=0.9pt},
  fab/.style={circle, draw=paperaccent!85!paperblack, fill=paperaccent!12, line width=0.9pt, align=center, inner sep=3pt, minimum size=1.1cm, font=\tiny},
  xcd/.style={rectangle, rounded corners=2pt, draw=paperquaternary!85!paperblack, fill=paperquaternary!12, line width=0.8pt, align=center, inner sep=2.5pt, font=\tiny},
  dbl/.style={{Latex[length=2mm]}-{Latex[length=2mm]}, line width=1.0pt, draw=paperaxis},
  spoke/.style={{Latex[length=1.3mm]}-{Latex[length=1.3mm]}, line width=0.5pt, draw=paperaxis!70},
  mem/.style={font=\tiny, text=papertertiary!40!paperblack},
  note/.style={font=\scriptsize, text=papermuted, align=center},
  panellabel/.style={font=\small\bfseries, text=paperaxis}
]

% ---------- H100 SXM ----------
\node[panellabel] at (2.4, 3.0) {H100 SXM};
\node[die, text width=2.9cm] (h100die) at (2.4, 1.95) {Compute die, monolithic\\989.5 TF BF16};
\draw[dbl] (h100die.south) to node[right, font=\scriptsize] {3.35 TB/s} (2.4, 0.82);
\draw[bar] (0.55, 0.22) rectangle (4.25, 0.82);
\foreach \x in {1.29, 2.03, 2.77, 3.51} \draw[papertertiary!85!paperblack, line width=0.6pt] (\x, 0.22) -- (\x, 0.82);
\node[mem] at (2.4, 0.52) {80 GB HBM3};
\node[note, font=\tiny] at (2.4, -0.28) {HBM3: 5 $\times$ 16 GB. 0.67 TB/s per stack};
\node[note] at (2.4, -1.02) {monolithic: no placement penalty\\$R_{\infty} = 10.2$M tok/s};

% ---------- B200 ----------
\node[panellabel] at (8.3, 3.0) {B200};
\node[die, text width=2.1cm] (b200d0) at (6.55, 1.95) {Die 0\\1125 TF BF16};
\node[die, text width=2.1cm] (b200d1) at (10.05, 1.95) {Die 1\\1125 TF BF16};
\draw[dbl, line width=1.3pt, draw=papersecondary] (b200d0.east) -- (b200d1.west);
\node[font=\tiny, text=papersecondary, align=center] at (8.3, 1.62) {NV-HBI\\10 TB/s};
\draw[dbl] (b200d0.south) to node[left, font=\scriptsize] {4 TB/s} (6.55, 0.82);
\draw[dbl] (b200d1.south) to node[right, font=\scriptsize] {4 TB/s} (10.05, 0.82);
\draw[bar] (5.35, 0.22) rectangle (7.75, 0.82);
\foreach \x in {5.95, 6.55, 7.15} \draw[papertertiary!85!paperblack, line width=0.6pt] (\x, 0.22) -- (\x, 0.82);
\node[mem] at (6.55, 0.52) {96 GB HBM3e};
\draw[bar] (8.85, 0.22) rectangle (11.25, 0.82);
\foreach \x in {9.45, 10.05, 10.65} \draw[papertertiary!85!paperblack, line width=0.6pt] (\x, 0.22) -- (\x, 0.82);
\node[mem] at (10.05, 0.52) {96 GB HBM3e};
\node[note, font=\tiny] at (8.3, -0.28) {HBM3e: 8 $\times$ 24 GB. 1 TB/s per stack};
\node[note] at (8.3, -1.02) {striped 8 TB/s: $R_{\infty} = 24.4$M tok/s\\pinned 4 TB/s: $R_{\infty} = 12.2$M tok/s ($2\times$)};

% ---------- MI300X ----------
\node[panellabel] at (13.7, 3.0) {MI300X};
\node[fab] (fabric) at (13.7, 1.35) {Infinity\\Fabric};
\foreach \i in {0,...,7} {
  \node[xcd] (x\i) at ({13.7+1.4*cos(45*\i+90)}, {1.35+1.02*sin(45*\i+90)}) {XCD \i\\163 TF};
}
\foreach \i in {0,...,7} {
  \draw[spoke] (fabric) -- (x\i);
}
\node[note, font=\tiny] at (13.7, -0.28) {8 XCDs. HBM3: 8 $\times$ 24 GB. 662.5 GB/s per stack};
\node[note] at (13.7, -1.02) {striped 5.3 TB/s: $R_{\infty} = 16.2$M tok/s\\pinned 0.66 TB/s: $R_{\infty} = 2.0$M tok/s ($8\times$)};

\end{tikzpicture}
\caption{Package topology of the three benchmarked accelerators \cite{nvidia2023h100, nvidia2024b200, amd2024mi300x}. H100 presents a monolithic die with uniform HBM access. B200 splits compute and HBM across two dies joined by NV-HBI at 10 TB/s, so a KV cache pinned to one die is served at 4 TB/s while striped pages reach the 8 TB/s aggregate. MI300X partitions compute across eight XCD chiplets with a 662.5 GB/s per-stack bandwidth slice; its on-package Infinity Fabric bandwidth is not publicly disclosed. The per-panel decode ceilings $R_{\infty}$ of Eq.~\eqref{eq:tokrate} quantify the placement penalty: pinning costs $2\times$ on B200 and $8\times$ on MI300X, making KV page placement a first-order performance decision.}
\label{fig:topology}
\end{figure}

The per-sequence decode rate at long context approaches a hardware ceiling obtained by dropping the weight term from Eq.~\eqref{eq:decode_bytes}.
\begin{equation}
R_{\infty} = \frac{\beta_{\mathrm{eff}}}{2 L n_{kv} d_{h} w_{kv}}.
\label{eq:tokrate}
\end{equation}
For Llama-3-70B in BF16 this ceiling is $10.2$ M tokens per second on H100, $24.4$ M on B200 with striped pages, $12.2$ M on B200 pinned to one die, $16.2$ M on MI300X striped, and $2.0$ M on MI300X pinned. Topology awareness is worth up to $8\times$ on the ceiling before any compression is applied.

\subsection{Crossover Context Lengths}

Equating the two terms of Eq.~\eqref{eq:decode_bytes} gives the traffic crossover, the context length at which per-step KV bytes equal per-step weight bytes.
\begin{equation}
s^{*}(b) = \frac{P \, w_{p}}{2 \, b \, L \, n_{kv} \, d_{h} \, w_{kv}}.
\label{eq:crossover_traffic}
\end{equation}
Table~\ref{tab:crossover} evaluates Eq.~\eqref{eq:crossover_traffic} for the Llama-3 family. The crossover falls hyperbolically in batch size. At the production batch sizes of $32$ to $256$ used by throughput-oriented serving, Llama-3-70B crosses at $13.4$k down to $1.7$k tokens, meaning long-context serving at scale is KV bound essentially always. Figure~\ref{fig:crossover} plots the three curves and marks the 128k operating point.

\begin{table}[t]
\centering
\caption{Traffic crossover context length $s^{*}(b)$ in tokens from Eq.~\eqref{eq:crossover_traffic}, BF16 weights and cache. Above the tabulated value, per-step KV traffic exceeds weight traffic.}
\label{tab:crossover}
\footnotesize
\begin{tabular}{lrrrrr}
\toprule
Model & $b{=}1$ & $b{=}8$ & $b{=}32$ & $b{=}128$ & $b{=}256$ \\
\midrule
Llama-3-8B & 122.5k & 15.3k & 3.8k & 957 & 479 \\
Llama-3-70B & 427.2k & 53.4k & 13.4k & 3.3k & 1.7k \\
Llama-3.1-405B & 1.57M & 196.2k & 49.0k & 12.3k & 6.1k \\
\bottomrule
\end{tabular}
\end{table}

A second crossover governs capacity rather than traffic. The context length at which the cache exhausts device memory, with weights co-resident, is
\begin{equation}
s_{\mathrm{cap}} = \frac{C_{\mathrm{hbm}} - P \, w_{p}}{2 L n_{kv} d_{h} w_{kv}}.
\label{eq:crossover_capacity}
\end{equation}
For Llama-3-8B in BF16, $s_{\mathrm{cap}}$ is $488$k tokens on one H100 and $1.34$M on one B200 or MI300X. For Llama-3-70B in BF16 the numerator is negative on H100, so the model does not fit at any context length, and on B200 or MI300X the limit is $159$k tokens. For Llama-3.1-405B the limit is negative on all three devices, forcing multi-device placement or weight quantization before context length even enters the discussion. Capacity, not bandwidth, binds first at the largest scale, which is precisely the regime where the tiering domain of Section~2 becomes necessary.

% Figure 9. Traffic crossover context length versus batch size, Llama-3 family.
% s*(b) from Eq. crossover_traffic. KV-dominated regime lies above each curve.
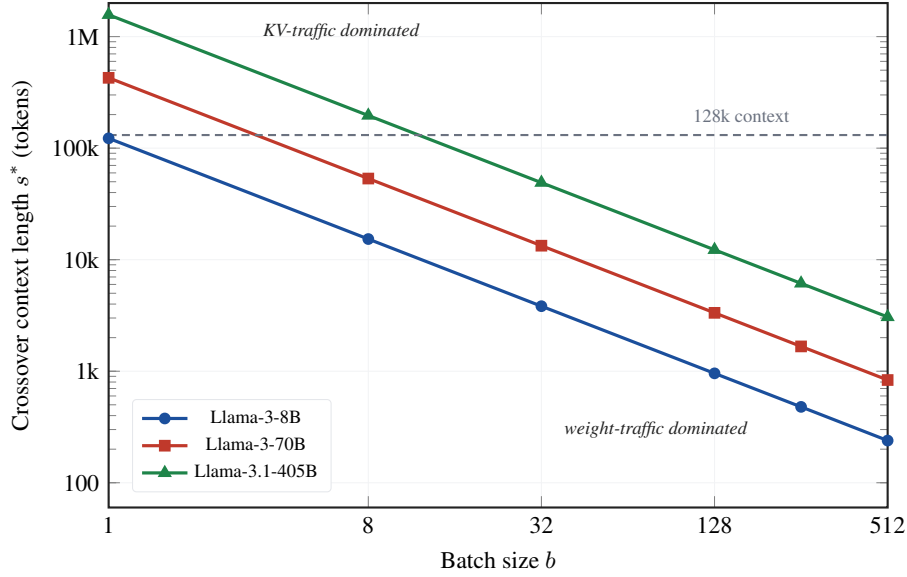
\begin{figure}[t]
\centering
\begin{tikzpicture}
\begin{loglogaxis}[
  width=0.72\linewidth,
  height=0.5\linewidth,
  xlabel={Batch size $b$},
  ylabel={Crossover context length $s^{*}$ (tokens)},
  xmin=1, xmax=512,
  ymin=60, ymax=2e6,
  xtick={1,8,32,128,512},
  xticklabels={1, 8, 32, 128, 512},
  ytick={1e2,1e3,1e4,1e5,1e6},
  yticklabels={100, 1k, 10k, 100k, 1M},
  grid=major,
  grid style={draw=papergrid!50, line width=0.3pt},
  axis line style={draw=paperaxis, thick},
  tick label style={font=\small},
  label style={font=\small},
  legend style={font=\scriptsize, draw=papergrid, fill=paperwhite, rounded corners=1.5pt},
  legend pos=south west,
  clip=false
]
\addplot[color=paperprimary, line width=1.2pt, mark=*, mark size=1.6pt] coordinates {
  (1,122528) (8,15316) (32,3829) (128,957) (256,479) (512,239)
};
\addplot[color=papersecondary, line width=1.2pt, mark=square*, mark size=1.6pt] coordinates {
  (1,427246) (8,53406) (32,13351) (128,3338) (256,1669) (512,834)
};
\addplot[color=papertertiary, line width=1.2pt, mark=triangle*, mark size=1.9pt] coordinates {
  (1,1569520) (8,196190) (32,49048) (128,12262) (256,6131) (512,3065)
};
\legend{Llama-3-8B, Llama-3-70B, Llama-3.1-405B}

% 128k reference line
\addplot[domain=1:512, samples=2, color=papermuted, densely dashed, line width=0.8pt] {131072};
\node[font=\scriptsize, text=papermuted, anchor=south west] at (axis cs:100, 1.4e5) {128k context};

% region labels
\node[font=\scriptsize\itshape, text=paperaxis, anchor=west] at (axis cs:3.2, 1.1e6) {KV-traffic dominated};
\node[font=\scriptsize\itshape, text=paperaxis] at (axis cs:80, 300) {weight-traffic dominated};

\end{loglogaxis}
\end{tikzpicture}
\caption{Traffic crossover context length $s^{*}(b)$ from Eq.~\eqref{eq:crossover_traffic}, where per-step KV bytes equal weight bytes. Above each curve the decode step is KV bound. Production batch sizes of 32 to 256 push the Llama-3-70B crossover from 13.4k down to 1.7k tokens, so most long-context serving already operates in the KV-bound regime. The dashed line marks 128k context.}
\label{fig:crossover}
\end{figure}

\subsection{Worked Bounds and the Speedup Structure of Compression}

In the memory-bound regime, step time is bytes over bandwidth, so the speedup from any technique that reduces KV bytes by a factor $c$ is the byte ratio alone.
\begin{equation}
S(c;\, s, b) = \frac{P w_{p} + 2 b L n_{kv} d_{h} w_{kv} s}{P w_{p} + 2 b L n_{kv} d_{h} (w_{kv}/c)\, s},
\qquad
\lim_{s \to 0} S = 1, \qquad
\lim_{s \to \infty} S = c.
\label{eq:speedup}
\end{equation}
These two limits explain most of the disagreement in reported speedups. In the weight-bound regime, compression is diluted by the weight term and its speedup vanishes toward one. In the KV-bound regime, the speedup converges to the compression factor exactly. Table~\ref{tab:worked_decode} instantiates Eq.~\eqref{eq:speedup} for Llama-3-70B on B200 at batch one with striped pages. At 128k context, below the crossover of $427.2$k, 2-bit quantization yields $1.25\times$. At 512k, above the crossover, the same technique yields $1.91\times$, and stacking 50 percent eviction reaches $2.05\times$. The identical algorithm differs in effectiveness by more than fifty percent depending only on which side of $s^{*}$ the workload sits. This is the analytical explanation for the inconsistent speedup claims cataloged in Section~1.

\begin{table}[t]
\centering
\caption{Worked roofline bounds for Llama-3-70B on B200 at batch one with striped pages, derived from Eqs.~\eqref{eq:decode_flops}, \eqref{eq:decode_bytes}, and~\eqref{eq:ai_general}. Derived quantities, not measurements. $T_{\mathrm{step}} = B/\beta$ with $\beta = 8.0$ TB/s; absolute times scale with realized bandwidth and ratios do not. Speedup relative to the BF16 baseline at the same context. The 128k block fits on one B200 ($181.9$ GB $\le 192$ GB). At 512k the BF16 and INT8 configurations exceed single-device capacity (307.8 and 223.9 GB against 192 GB) and are realized by sharding across two devices; the INT4 and 2-bit rows fit on one B200, and all speedup ratios are invariant to the shard count.}
\label{tab:worked_decode}
\footnotesize
\begin{tabular}{lrrrrrr}
\toprule
Configuration & $w_{kv}$ & KV (GB) & $T_{\mathrm{step}}$ (ms) & tok/s & Speedup & AI \\
\midrule
\multicolumn{7}{l}{\emph{Context $s = 128$k, below crossover $s^{*} = 427.2$k.}} \\
BF16 baseline & 16 bit & 41.9 & 22.7 & 44.0 & 1.00$\times$ & 2.6 \\
INT8 cache & 8 bit & 21.0 & 20.1 & 49.7 & 1.13$\times$ & 3.0 \\
INT4 cache & 4 bit & 10.5 & 18.8 & 53.2 & 1.21$\times$ & 3.2 \\
2-bit cache & 2 bit & 5.2 & 18.2 & 55.1 & 1.25$\times$ & 3.3 \\
2-bit plus $r = 0.5$ eviction & 2 bit & 2.6 & 17.8 & 56.1 & 1.28$\times$ & 3.3 \\
\midrule
\multicolumn{7}{l}{\emph{Context $s = 512$k, above crossover.}} \\
BF16 baseline & 16 bit & 167.8 & 38.5 & 26.0 & 1.00$\times$ & 4.8 \\
INT8 cache & 8 bit & 83.9 & 28.0 & 35.7 & 1.37$\times$ & 6.6 \\
INT4 cache & 4 bit & 41.9 & 22.7 & 44.0 & 1.69$\times$ & 8.1 \\
2-bit cache & 2 bit & 21.0 & 20.1 & 49.7 & 1.91$\times$ & 9.2 \\
2-bit plus $r = 0.5$ eviction & 2 bit & 10.5 & 18.8 & 53.2 & 2.05$\times$ & 9.8 \\
\bottomrule
\end{tabular}
\end{table}

Three regimes follow. Below $s^{*}$, weight traffic dominates and KV compression yields bounded speedups approaching one. Above $s^{*}$, KV traffic dominates and byte reductions translate one-to-one into speedup until quality constraints bind. Beyond $s_{\mathrm{cap}}$, no placement exists at all and tiering or multi-device parallelism becomes mandatory. Every technique in Section~2 can be located on this map, and Section~4 measures a representative of each domain against it.
% ======================================================================
%  sections/sec_formal_statements.tex
%  Formal statements and proofs of the Section 3 analytical results.
%
%  Wiring: \input{sections/sec_formal_statements} at the END of
%  Section 3, immediately after the closing paragraph "Three regimes
%  follow ..." and before Section 4. Numbering is automatic: the
%  subsection becomes 3.6 and the statements become Proposition 3.1
%  and so on. No main-text section or table number moves.
%
%  Preamble additions (once, in main.tex):
%    \usepackage{amsthm}
%    \theoremstyle{plain}
%    \newtheorem{proposition}{Proposition}[section]
%    \newtheorem{corollary}[proposition]{Corollary}
%    \theoremstyle{remark}
%    \newtheorem{remark}[proposition]{Remark}
% ======================================================================

\subsection{Formal Statements of the Analytical Results}
\label{sec:formal}

The derivations of this section are accounting identities over FLOPs and bytes rather than fits to measurements, and we state their content formally. Propositions~\ref{prop:ai_limits}--\ref{prop:beta_eff} assume the cost model of Eqs.~\eqref{eq:decode_flops} and~\eqref{eq:decode_bytes}, the model constants of Table~\ref{tab:model_kv}, the device constants of Table~\ref{tab:hardware}, and the placement model of Eq.~\eqref{eq:beta_eff}. Appendix~\ref{app:derivations} carries the supporting arithmetic in full.

\begin{proposition}[Decode intensity limits]
\label{prop:ai_limits}
Let $\mathrm{AI}(s,b) = F(s,b)/B(s,b)$ as in Eq.~\eqref{eq:ai_general} and let $g = n_{q}/n_{kv}$. Then
\begin{enumerate}[leftmargin=1.5em]
  \item $\lim_{s \to 0} \mathrm{AI}(s,b) = 2b/w_{p}$ (weight regime);
  \item $\lim_{s \to \infty} \mathrm{AI}(s,b) = 2g/w_{kv}$ (KV floor), with the batch size canceling exactly;
  \item for fixed $b$, the map $s \mapsto \mathrm{AI}(s,b)$ is strictly increasing if $b < b^{\dagger}$, constant if $b = b^{\dagger}$, and strictly decreasing if $b > b^{\dagger}$, where $b^{\dagger} = g w_{p}/w_{kv}$ is the fixed point of Eq.~\eqref{eq:fixed_point};
  \item for fixed $s > 0$, the map $b \mapsto \mathrm{AI}(s,b)$ is strictly increasing with least upper bound
  \[
  \sup_{b} \mathrm{AI}(s,b) = \frac{2g}{w_{kv}} + \frac{P}{L n_{kv} d_{h} w_{kv} s}.
  \]
\end{enumerate}
\end{proposition}

\begin{proof}
Write $W = P w_{p}$, $a = 4 L n_{q} d_{h}$, and $k = 2 L n_{kv} d_{h} w_{kv}$, so that $\mathrm{AI}(s,b) = b(2P + a s)/(W + b k s)$ with all constants positive. Claims (i) and (ii) follow by dividing numerator and denominator by $s$ and taking limits. For (iii), differentiating in $s$ gives
\[
\frac{\partial \mathrm{AI}}{\partial s} = \frac{b\,(aW - 2Pbk)}{(W + bks)^{2}} = \frac{4 P L n_{kv} d_{h}\, b\, (g w_{p} - b w_{kv})}{(W + bks)^{2}},
\]
whose sign is the sign of $g w_{p} - b w_{kv}$, hence positive, zero, or negative as $b < b^{\dagger}$, $b = b^{\dagger}$, or $b > b^{\dagger}$. For (iv),
\[
\frac{\partial \mathrm{AI}}{\partial b} = \frac{(2P + a s)\, W}{(W + bks)^{2}} > 0,
\qquad
\lim_{b \to \infty} \mathrm{AI}(s,b) = \frac{2P + a s}{k s} = \frac{a}{k} + \frac{2P}{k s} = \frac{2g}{w_{kv}} + \frac{P}{L n_{kv} d_{h} w_{kv} s},
\]
and the limit is a supremum by strict monotonicity.
\end{proof}

\begin{remark}[Span of the batch lever]
\label{rem:batch_span}
For Llama-3-70B in BF16 at 128k context, Proposition~\ref{prop:ai_limits}(ii) and~(iv) bound the intensity of any decode configuration between $8.0$ and $11.3$ FLOP/B. The batch lever moves decode intensity by at most $42$ percent at this context, and both ends sit $25\times$ to $35\times$ below the B200 ridge of $281$ FLOP/B. This is the formal content of the claim that no batch size makes long-context decode compute bound.
\end{remark}

\begin{proposition}[Traffic crossover]
\label{prop:crossover}
Let the KV share of per-step traffic be $\eta(s,b) = 2 b L n_{kv} d_{h} w_{kv} s / B(s,b)$ and let $s^{*}(b)$ be as in Eq.~\eqref{eq:crossover_traffic}. Then
\begin{enumerate}[leftmargin=1.5em]
  \item $\eta(s,b) = s/(s + s^{*}(b))$; in particular $\eta = \tfrac{1}{2}$ exactly at $s = s^{*}(b)$, and $\eta$ is strictly increasing in $s$ and in $b$;
  \item $s^{*}(b) = s^{*}(1)/b$: the crossover is hyperbolic in batch size, so batching accelerates the onset of the KV-bound regime rather than postponing it;
  \item the weight-side intensity reaches the ridge $I^{*}$ only at $b = I^{*} w_{p}/2$, and
  \[
  s^{*}\!\left(\tfrac{I^{*} w_{p}}{2}\right) = \frac{P}{I^{*} L n_{kv} d_{h} w_{kv}}.
  \]
\end{enumerate}
\end{proposition}

\begin{proof}
Dividing the numerator and denominator of $\eta$ by $W = P w_{p}$ and using $2 b L n_{kv} d_{h} w_{kv} / W = 1/s^{*}(b)$ from Eq.~\eqref{eq:crossover_traffic} gives $\eta = (s/s^{*})/(1 + s/s^{*}) = s/(s + s^{*})$. Monotonicity in $s$ follows from $\partial \eta / \partial s = s^{*}/(s + s^{*})^{2} > 0$, and monotonicity in $b$ follows since $s^{*}(b)$ is strictly decreasing. Claim (ii) is the functional form of Eq.~\eqref{eq:crossover_traffic}. For (iii), the weight-regime limit of Proposition~\ref{prop:ai_limits}(i) equals $I^{*}$ at $b = I^{*} w_{p}/2$; substituting into Eq.~\eqref{eq:crossover_traffic} and canceling $w_{p}$ gives the stated value, and the constants of Tables~\ref{tab:model_kv} and~\ref{tab:hardware} give the numerical instantiation: $1.4$k tokens for Llama-3-70B on H100, so the batch size at which weights could become compute bound already places every longer context in the KV-bound regime.
\end{proof}

\begin{proposition}[Speedup structure of byte reduction]
\label{prop:speedup_ceiling}
In the memory-bound regime $T_{\mathrm{step}} = B(s,b)/\beta$, let $c > 1$ be the factor by which a technique reduces KV bytes and let $S(c; s, b)$ be the step-time speedup of Eq.~\eqref{eq:speedup}. Then
\begin{enumerate}[leftmargin=1.5em]
  \item (closed form in the crossover) $\displaystyle S(c; s, b) = \frac{c\,(s + s^{*}(b))}{s + c\, s^{*}(b)}$;
  \item (ceiling) $1 < S < c$, $\lim_{s \to 0} S = 1$, and $\lim_{s \to \infty} S = c$: compression speedup converges to the byte factor exactly, and only in the KV-bound regime;
  \item (monotonicity) $S$ is strictly increasing in $s$ and in $b$;
  \item (sub-multiplicative composition) for byte factors $c_{1}, c_{2} > 1$,
  \[
  S(c_{1} c_{2};\, s, b) \;\le\; S(c_{1};\, s, b)\; S(c_{2};\, s, b),
  \]
  with equality only at the regime boundaries $s \to 0$ and $s \to \infty$.
\end{enumerate}
\end{proposition}

\begin{proof}
Let $W = P w_{p}$ and $K = 2 b L n_{kv} d_{h} w_{kv} s$, so that $K/W = s/s^{*}(b)$ by Eq.~\eqref{eq:crossover_traffic}. Then $S = (W + K)/(W + K/c)$, and dividing numerator and denominator by $W$ and clearing $c\, s^{*}$ gives (i). For (ii), $S - 1 = K(c-1)/(cW + K) > 0$ and $c - S = cW(c-1)/(cW + K) > 0$; the limits follow from $K/W \to 0$ and $K/W \to \infty$. For (iii), differentiating the closed form in $s$ gives $\partial S/\partial s = c\, s^{*} (c - 1)/(s + c s^{*})^{2} > 0$, and $S$ inherits monotonicity in $b$ from $s^{*}(b)$, since $\partial S/\partial s^{*} = c s (1 - c)/(s + c s^{*})^{2} < 0$ while $s^{*}(b)$ is strictly decreasing. For (iv), applying the closed form to each factor and to the product,
\[
\frac{S(c_{1} c_{2})}{S(c_{1})\, S(c_{2})} = \frac{(s + c_{1} s^{*})(s + c_{2} s^{*})}{(s + s^{*})(s + c_{1} c_{2} s^{*})} \le 1,
\]
because the denominator minus the numerator expands to $s s^{*} (c_{1} - 1)(c_{2} - 1) \ge 0$, with equality iff $s s^{*} = 0$ or one factor equals one.
\end{proof}

\begin{remark}[Consistency with the benchmark]
\label{rem:speedup_consistency}
Proposition~\ref{prop:speedup_ceiling}(i) reproduces every derived speedup in the paper: at W1 ($s = 128$k, $b = 1$, $s^{*} = 427.2$k) it gives $1.25\times$ for $c = 8$, and at W2 ($s = 128$k, $b = 32$, $s^{*} = 13.4$k) it gives $4.82\times$ for $c = 8$ and $6.6\times$ for the composed $c = 16$ configuration. Claim (iv) is quantitative in Table~\ref{tab:worked_decode}: at $s = 512$k the factors $c = 8$ and $c = 2$ predict $1.91 \times 1.37 = 2.63\times$ if multiplied naively, while the joint evaluation gives the correct $2.05\times$. Stacked methods must be evaluated at their joint byte factor, never by multiplying per-method speedups.
\end{remark}

\begin{proposition}[Effective-bandwidth bounds]
\label{prop:beta_eff}
In the placement model of Eq.~\eqref{eq:beta_eff} with $D$ partitions, local bandwidth $\beta_{p} = \beta_{\mathrm{agg}}/D$, and fabric bandwidth $\beta_{f}$, the effective bandwidth at locality fraction $\gamma \in (0,1)$ is
\[
\beta_{\mathrm{eff}}(\gamma) = \min\!\left(\frac{\beta_{p}}{\gamma},\; \frac{\beta_{f}}{1-\gamma}\right).
\]
The placement extremes give $\beta_{\mathrm{eff}} = \beta_{p}$ for pinned placement ($\gamma = 1$) and $\beta_{\mathrm{eff}} = \min\!\left(\beta_{\mathrm{agg}},\, \tfrac{D}{D-1}\beta_{f}\right)$ for striped placement ($\gamma = 1/D$). In particular, striping recovers the full aggregate bandwidth whenever $\beta_{f} \ge (1 - 1/D)\,\beta_{\mathrm{agg}}$, and pinning never exceeds the per-partition rate.
\end{proposition}

\begin{proof}
The expression is $x/T_{\mathrm{read}}$ with $T_{\mathrm{read}}$ from Eq.~\eqref{eq:beta_eff}. Pinned placement has $\gamma = 1$ and all bytes local, so $T_{\mathrm{read}} = x/\beta_{p}$. Striped placement has $\gamma = 1/D$, giving $\beta_{p}/\gamma = D \beta_{p} = \beta_{\mathrm{agg}}$ and $\beta_{f}/(1 - \gamma) = \tfrac{D}{D-1}\beta_{f}$, and the minimum of the two is the stated bound. The striping condition follows by requiring the second argument to be at least $\beta_{\mathrm{agg}}$.
\end{proof}

\begin{remark}[Scope of the placement model]
\label{rem:beta_scope}
Equation~\eqref{eq:beta_eff} charges remote reads to the fabric alone and does not model contention for the remote partition's HBM. The model is therefore exact for pinned placement and an upper bound that favors remote reads for striped placement. The B200 striped evaluation of Section~3 saturates at the aggregate rate, where the omission is immaterial, and the MI300X striped figures carry the stated assumption that the on-package fabric does not bind. The first-order serving guidance, stripe pages and never serve per-step reads from a remote tier, is unaffected.
\end{remark}

\begin{corollary}[Decode-rate ceiling]
\label{cor:tokrate}
Let $R_{\infty} = \beta_{\mathrm{eff}}/(2 L n_{kv} d_{h} w_{kv})$ as in Eq.~\eqref{eq:tokrate}. Then the decode throughput satisfies $s \cdot b / T_{\mathrm{step}} \to R_{\infty}$ as $s \to \infty$: the product of context length and token rate is a hardware constant. Equivalently, the long-context token rate at context $s$ is $R_{\infty}/s$ before weight traffic is accounted. For Llama-3-70B in BF16, $R_{\infty}$ is $10.2$ M tokens per second on H100, $24.4$ M on B200 striped, $12.2$ M on B200 pinned, $16.2$ M on MI300X striped, and $2.0$ M on MI300X pinned.
\end{corollary}

\begin{proof}
By Eqs.~\eqref{eq:decode_bytes} and~\eqref{eq:tokrate}, $s\, b / T_{\mathrm{step}} = s b\, \beta_{\mathrm{eff}} / (P w_{p} + 2 b L n_{kv} d_{h} w_{kv} s)$, and dividing numerator and denominator by $s$ sends the weight term to zero. The numerical values are the device bandwidths of Table~\ref{tab:hardware} divided by the $327{,}680$ bytes per token of Table~\ref{tab:model_kv}, worked in Appendix~\ref{app:constants}.
\end{proof} 
\section{Comparative SoK Analysis Across the Five Compression Domains}
\label{sec:benchmark}

This section evaluates one representative method per taxonomy domain on common workloads. The protocol fixes the model (Llama-3-70B, Table~\ref{tab:model_kv}), the hardware (B200 with striped pages, Table~\ref{tab:hardware}), and the context length (128k tokens) for all derived columns, and it imports quality numbers only as within-paper reported values with their benchmark named. Two evidence classes appear throughout and are never mixed within a cell. Derived values come from Eqs.~\eqref{eq:decode_bytes} and~\eqref{eq:speedup}. Reported values come from the cited papers.

\subsection{Workload and Metric Definitions}

Table~\ref{tab:workloads} defines four workloads spanning the operating points that production reasoning serving actually exercises.

\begin{table}[t]
\centering
\caption{Benchmark workloads. W1 and W2 isolate the batch-size lever. W3 isolates prefix sharing. W4 isolates reuse across sessions.}
\label{tab:workloads}
\footnotesize
\begin{tabular}{l>{\raggedright\arraybackslash}p{3.6cm}ccl}
\toprule
ID & Pattern & $b$ & $s$ & Binding constraint \\
\midrule
W1 & Single-sequence decode & 1 & 128k & Latency, weight-dominated \\
W2 & Batched decode & 32 & 128k & Throughput, KV-dominated \\
W3 & 8 requests, 96k shared prefix plus 32k unique & 8 & 128k & Capacity and TTFT \\
W4 & Reuse after 10-minute gap & 1 & 128k & Recompute versus fetch \\
\bottomrule
\end{tabular}
\end{table}

Four metrics accompany each method. The byte factor $c$ divides Eq.~\eqref{eq:decode_bytes} traffic. The footprint is KV gigabytes at 128k context per sequence. Decode speedup is Eq.~\eqref{eq:speedup} evaluated at W1 and W2. Quality delta is the accuracy or perplexity change reported by the method's authors on their stated benchmark. The speedup columns are derived bounds. They assume perfect kernel efficiency and full bandwidth realization and ignore dequantization overhead, which Section~5 discusses; absolute times move with realized bandwidth, ratios do not. W2 at $b = 32$ and the eight-request aggregate of W3 ($255$ GB with weights) exceed single-device capacity and are evaluated as traffic-ratio regimes; the ratios are capacity-independent.

\subsection{The SoK Comparison Matrix}

Table~\ref{tab:sok_matrix} presents the matrix. Three readings matter. First, the same 2-bit quantization yields $1.25\times$ at W1 and $4.82\times$ at W2. The algorithm did not change. The regime did, exactly as Eq.~\eqref{eq:speedup} predicts. Second, the lossless domains move feasibility rather than decode speed. Paging converts stranded memory into usable batch slots, and prefix sharing divides stored bytes by $\sigma$ without touching per-step traffic. Third, tiering barely changes decode at all. It converts a 25-second recompute into a sub-second fetch on reuse, an economic gain that no decode-bound metric captures.

\begin{table}[t]
\centering
\caption{SoK comparison matrix. Llama-3-70B, B200 with striped pages, 128k context. Byte factors and footprints are exact from Eq.~\eqref{eq:kv_bytes_unified}. Decode speedups are derived from Eq.~\eqref{eq:speedup}. Quality deltas are within-paper reported values on the benchmarks named and are not cross-comparable. $^{\dagger}$Total stored across the 8 requests of W3, versus 335.5 GB without sharing.}
\label{tab:sok_matrix}
\footnotesize
\begin{tabularx}{\linewidth}{>{\raggedright\arraybackslash}p{2.2cm}crrcc>{\raggedright\arraybackslash}X}
\toprule
Method & $c$ & KV (GB) & W1 & W2 & Lossless & Quality delta (as reported) \\
\midrule
BF16 paged baseline \cite{kwon2023vllm} & 1$\times$ & 41.9 & 1.00$\times$ & 1.00$\times$ & --- & reference \\
KIVI, 2-bit \cite{liu2024kivi} & 8$\times$ & 5.2 & 1.25$\times$ & 4.82$\times$ & No & $+0.03$ to $0.3$ ppl. LongBench $\approx$ baseline \\
KVQuant, 3-bit \cite{hooper2024kvquant} & 5.3$\times$ & 7.9 & 1.23$\times$ & 3.78$\times$ & No & $<0.1$ ppl. LongBench near baseline \\
H2O, $r{=}0.5$ \cite{zhang2023h2o} & 2$\times$ & 21.0 & 1.13$\times$ & 1.83$\times$ & No & 97 to 99\% task retention reported \\
SnapKV, $r{=}0.25$ \cite{li2024snapkv} & 4$\times$ & 10.5 & 1.21$\times$ & 3.12$\times$ & No & 96 to 98\% LongBench retention reported \\
PagedAttention \cite{kwon2023vllm} & 1$\times$ & 41.9 & 1.00$\times$ & 1.00$\times$ & Yes & None. Waste cut from 60 to 80\% to $<4$\% \\
RadixAttention, W3 \cite{zheng2024sglang} & $\sigma{=}2.91$ & 115.3$^{\dagger}$ & 1.00$\times$ & 1.00$\times$ & Yes & None. Derived TTFT $2.8\times$ lower on warm W3 hits (App.~A) \\
CacheGen tiering \cite{liu2024cachegen} & codec 4$\times$ & 41.9 & 1.00$\times$ & 1.00$\times$ & Near & Negligible loss reported. W4 fetch 0.16 s versus 25.0 s recompute \\
\bottomrule
\end{tabularx}
\end{table}

\subsection{Precision Reduction Results}

Figure~\ref{fig:quantization} shows the footprint ladder and the reported quality deltas. Three findings hold across studies. Eight-bit storage is free. Four-bit storage is near-free once keys are quantized per channel and outliers are isolated \cite{hooper2024kvquant}. Two-bit storage holds only with structure-aware layouts and a full-precision residual window \cite{liu2024kivi}, while naive round-to-nearest at 2-bit collapses. The roofline gives the system-side meaning of the ladder. Each halving of $w_{kv}$ doubles the long-context intensity floor of Eq.~\eqref{eq:ai_floor}, from $8$ FLOP/B at BF16 to $64$ FLOP/B at 2-bit, and delivers the full compression factor as decode speedup once $s > s^{*}(b)$.

% Figure 4. KV quantization. Footprint ladder at 128k with reported quality deltas.
\begin{figure}[t]
\centering
\begin{tikzpicture}
\begin{axis}[
  paper bar plot,
  width=0.72\linewidth,
  height=0.48\linewidth,
  ylabel={KV footprint at 128k context (GB)},
  xlabel={Bits per stored element $w_{q}$},
  symbolic x coords={16, 8, 4, 3, 2},
  xtick=data,
  tick align=inside,
  ymin=0, ymax=50,
  ytick={0, 10, 20, 30, 40}
]
\addplot coordinates {(16,41.9) (8,21.0) (4,10.5) (3,7.9) (2,5.2)};
% bar-top labels: footprint, reported quality delta, derived W2 speedup (Table sok_matrix)
\node[font=\scriptsize, text=paperaxis, anchor=south, align=center] at (axis cs:16, 42.5) {41.9 GB\\[0.5pt] {\tiny\textcolor{papermuted}{reference}}};
\node[font=\scriptsize, text=paperaxis, anchor=south, align=center] at (axis cs:8, 21.6) {21.0 GB\\[0.5pt] {\textcolor{papertertiary!80!black}{$\Delta$ppl $\approx 0$}}\\[0.5pt] {\tiny\textcolor{papermuted}{W2 $1.83\times$}}};
\node[font=\scriptsize, text=paperaxis, anchor=south, align=center] at (axis cs:4, 11.1) {10.5 GB\\[0.5pt] {\textcolor{papertertiary!80!black}{$\Delta$ppl $+0.05$}}\\[0.5pt] {\tiny\textcolor{papermuted}{W2 $3.12\times$}}};
\node[font=\scriptsize, text=paperaxis, anchor=south, align=center] at (axis cs:3, 8.5) {7.9 GB\\[0.5pt] {\textcolor{papertertiary!80!black}{$+0.03$--$0.1$ ppl}}\\[0.5pt] {\tiny\textcolor{papermuted}{W2 $3.78\times$}}};
\node[font=\scriptsize, text=paperaxis, anchor=south, align=center] at (axis cs:2, 5.8) {5.2 GB\\[0.5pt] {\textcolor{papersecondary}{$+0.03$--$0.3$}}\\[0.5pt] {\textcolor{papersecondary}{naive $>1$ ppl}}\\[0.5pt] {\tiny\textcolor{papermuted}{W2 $4.82\times$}}};
\end{axis}
\end{tikzpicture}
\caption{KV footprint at 128k context per quantization level, Llama-3-70B, with quality deltas as reported by KIVI \cite{liu2024kivi} and KVQuant \cite{hooper2024kvquant}. Eight-bit storage is free. Four-bit is near-free with per-channel keys and outlier isolation. Two-bit holds only for structure-aware layouts. Naive round-to-nearest at 2-bit loses more than one perplexity point. Each halving of $w_{q}$ doubles the intensity floor of Eq.~\eqref{eq:ai_floor}. Bar labels give the footprint, the reported quality delta, and the derived W2 decode speedup of Eq.~\eqref{eq:speedup} at byte factor $c = 16/w_{q}$, matching Table~\ref{tab:sok_matrix}.}
\label{fig:quantization}
\end{figure}
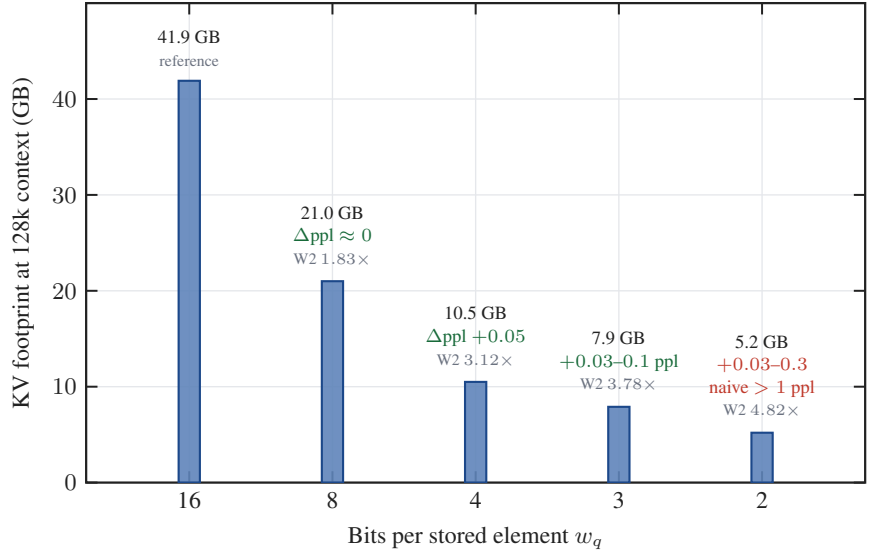

\subsection{Token Eviction Results}

Figure~\ref{fig:eviction} shows the retention-quality structure reported by eviction studies. Aggregate benchmarks such as LongBench tolerate aggressive eviction, with reported retention of 96 to 99 percent at $r = 0.25$ for heavy-hitter and observation-window policies \cite{zhang2023h2o, li2024snapkv}. Retrieval-style probes degrade discontinuously, because the evicted middle of the context is exactly where the answer lives. A single scalar quality budget is therefore insufficient. Eviction policies must be evaluated against the position sensitivity of the workload. Eviction composes with quantization, and Table~\ref{tab:worked_decode} showed that the combined byte factor multiplies.

% Figure 5. Token eviction. Retention ratio versus reported relative quality.
\begin{figure}[t]
\centering
\begin{tikzpicture}
\begin{axis}[
  width=0.72\linewidth,
  height=0.48\linewidth,
  xlabel={Retention ratio $r$},
  ylabel={Relative quality (\% of baseline)},
  xmin=0, xmax=1.1,
  ymin=0, ymax=110,
  xtick={0,0.1,0.25,0.5,0.75,1},
  tick align=inside,
  grid=major,
  grid style={draw=papergrid},
  axis line style={draw=paperaxis, thick},
  tick label style={font=\small},
  tick style={draw=black, line width=0.5pt},
  label style={font=\small},
  legend style={font=\scriptsize, draw=papergrid, fill=paperwhite, rounded corners=1.5pt},
  legend pos=south east,
  clip=false
]
\addplot[color=paperprimary, line width=1.2pt, mark=*, mark size=1.8pt] coordinates {
  (1,100) (0.5,99.5) (0.25,98.0) (0.1,96.0)
};
\addplot[color=papersecondary, line width=1.2pt, mark=square*, mark size=1.8pt] coordinates {
  (1,100) (0.5,99.0) (0.25,96.5) (0.1,92.5)
};
\addplot[color=paperquaternary, line width=1.2pt, mark=triangle*, mark size=3.0pt, densely dashed] coordinates {
  (1,100) (0.5,55) (0.25,25) (0.1,5)
};
\legend{SnapKV on LongBench (reported), H2O on LongBench (reported), Naive window eviction on needle retrieval}
\node[font=\scriptsize, text=papermuted, align=left, anchor=west] (retnote) at (axis cs:0.04, 72) {retrieval probes collapse\\when the middle is evicted};
\draw[draw=papermuted!70, line width=0.4pt, -{Latex[length=1.3mm]}] (retnote.east) to[out=0, in=160] (axis cs:0.63, 66);
% top-row footprint labels aligned with the x gridlines; drawn as nodes because
% pgfplots ignores tick position inside extra tick styles
\node[font=\tiny, text=papermuted, anchor=south] at (axis cs:0.1, 110.5) {4.2 GB};
\node[font=\tiny, text=papermuted, anchor=south] at (axis cs:0.25, 110.5) {10.5 GB};
\node[font=\tiny, text=papermuted, anchor=south] at (axis cs:0.5, 110.5) {21.0 GB};
\node[font=\tiny, text=papermuted, anchor=south] at (axis cs:1, 110.5) {41.9 GB};
\end{axis}
\end{tikzpicture}
\caption{Reported retention-quality behavior of token eviction. Aggregate long-context benchmarks tolerate $r = 0.25$ with 96 to 99 percent retention for heavy-hitter and observation-window policies \cite{zhang2023h2o, li2024snapkv}. Position-sensitive retrieval degrades discontinuously under window eviction \cite{xiao2023streamingllm}. Curves are representative within-paper values and are task dependent. Top-row labels give the KV footprint at 128k context for Llama-3-70B, $41.9\,r$ GB from Eq.~\eqref{eq:kv_bytes_unified}.}
\label{fig:eviction}
\end{figure}
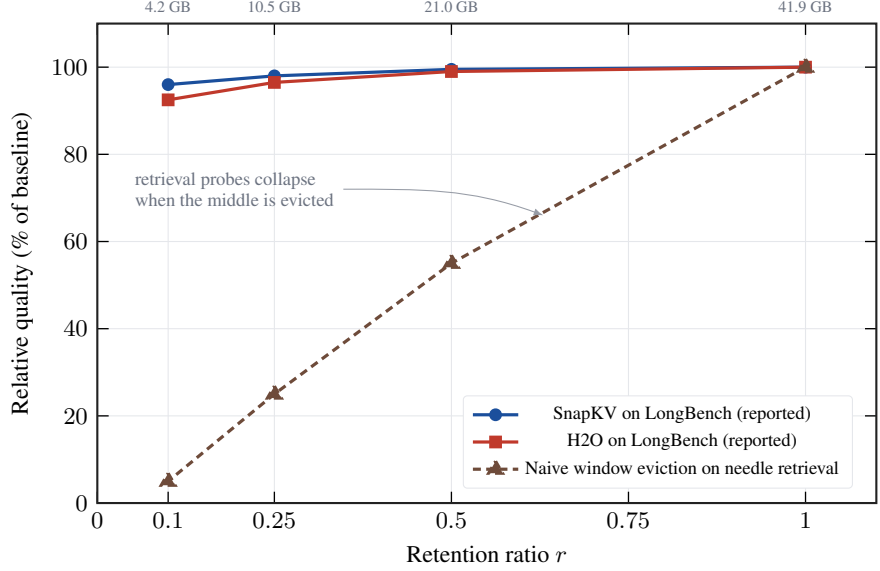

\subsection{Paging, Prefix Sharing, and Tiering Results}

Figure~\ref{fig:paging} shows the capacity effect of paging. Reservation-based allocation strands 60 to 80 percent of KV memory in over-reservation and fragmentation, while block-level allocation holds waste under 4 percent \cite{kwon2023vllm}. The gain appears as a larger feasible batch size rather than lower per-step bytes, which is why the W1 and W2 columns of Table~\ref{tab:sok_matrix} show $1.00\times$ for paging while production throughput roughly doubles to quadruples \cite{kwon2023vllm}.

% Figure 6. Paging versus reservation. Share of allocated KV memory by category.
\begin{figure}[t]
\centering
\begin{tikzpicture}
\begin{axis}[
  paper bar plot,
  width=0.66\linewidth,
  height=0.46\linewidth,
  ylabel={Share of allocated KV memory (\%)},
  xmin=0.4, xmax=2.6,
  xtick={1,2},
  xticklabels={Reservation, Paged $\tau{=}16$},
  ymin=0, ymax=112,
  tick align=inside,
  tick style={draw=black, line width=0.5pt},
  legend style={font=\small, at={(0.5,-0.1)}, anchor=north, legend columns=3}
]
% stacked columns drawn explicitly; segment placement must not depend on bar-style defaults
\fill[paperprimary!25, draw=paperprimary!85!black, line width=0.8pt] (axis cs:0.84,0) rectangle (axis cs:1.16,25);
\fill[papersecondary!28, draw=papersecondary!85!black, line width=0.8pt] (axis cs:0.84,25) rectangle (axis cs:1.16,70);
\fill[paperquaternary!32, draw=paperquaternary!85!black, line width=0.8pt] (axis cs:0.84,70) rectangle (axis cs:1.16,100);
\fill[paperprimary!25, draw=paperprimary!85!black, line width=0.8pt] (axis cs:1.84,0) rectangle (axis cs:2.16,96.5);
\fill[paperquaternary!32, draw=paperquaternary!85!black, line width=0.8pt] (axis cs:1.84,96.5) rectangle (axis cs:2.16,100);
\addlegendimage{area legend, fill=paperprimary!25, draw=paperprimary!85!black, line width=0.8pt}
\addlegendentry{ Useful KV data }
\addlegendimage{area legend, fill=papersecondary!28, draw=papersecondary!85!black, line width=0.8pt}
\addlegendentry{ Over-reservation }
\addlegendimage{area legend, fill=paperquaternary!32, draw=paperquaternary!85!black, line width=0.8pt}
\addlegendentry{ Fragmentation}
% in-segment share labels
\node[font=\scriptsize, text=paperaxis] at (axis cs:1, 12.5) {25\%};
\node[font=\scriptsize, text=paperaxis] at (axis cs:1, 47.5) {45\%};
\node[font=\scriptsize, text=paperaxis] at (axis cs:1, 85) {30\%};
\node[font=\scriptsize, text=paperaxis] at (axis cs:2, 48) {96.5\%};
\node[font=\scriptsize, text=paperquaternary!70!black, anchor=west] at (axis cs:2.2, 98.2) {3.5\%};
\end{axis}
\end{tikzpicture}
\caption{Allocation efficiency of KV memory. Reservation-based serving strands 60 to 80 percent of cache memory in over-reservation and fragmentation \cite{kwon2023vllm}. The category split of 25, 45, and 30 percent is illustrative within that reported range. PagedAttention bounds waste to the partially filled final block, measured below 4 percent \cite{kwon2023vllm}.}
\label{fig:paging}
\end{figure}
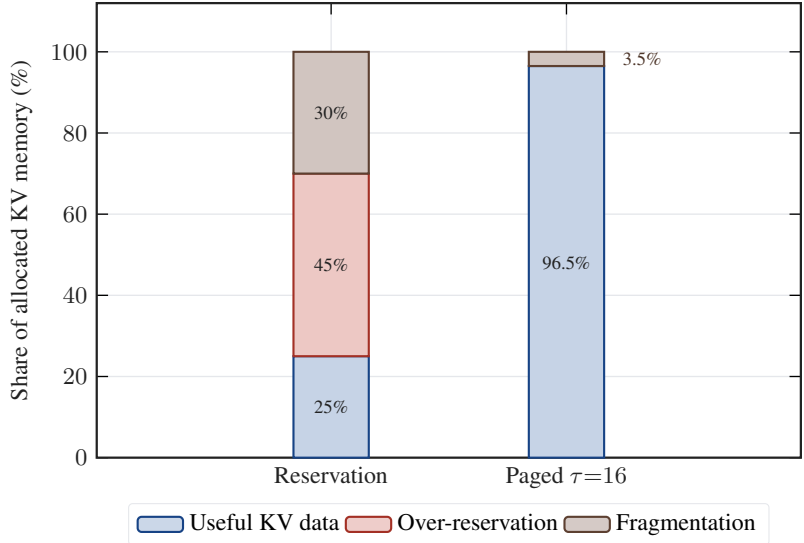

Figure~\ref{fig:prefix} shows the TTFT effect of prefix sharing. Time to first token falls with hit rate, linearly in the weight-matmul term and as $1 - h^{2}$ in the attention term (Appendix~\ref{app:derivations}): from $25.0$ s cold to a derived $3.7$ s at the 90 percent hit rates reported on sharing-heavy traces \cite{zheng2024sglang}. On W3 the shared 96k prefix is encoded once for 8 requests. A request that hits the warm prefix encodes its 32k suffix against the resident cache, a derived $8.8$ s versus $25.0$ s cold, or $2.8\times$ lower TTFT, and stored bytes fall $2.91\times$, with zero quality cost.

% Figure 7. Prefix cache hit rate versus time to first token, 128k context, B200.
\begin{figure}[t]
\centering
\begin{tikzpicture}
\begin{axis}[
  width=0.72\linewidth,
  height=0.48\linewidth,
  xlabel={Prefix cache hit rate},
  ylabel={Time to first token (s)},
  xmin=-0.03, xmax=1,
  ymin=0, ymax=28,
  tick align=inside,
  tick style={draw=black, line width=0.5pt},
  grid=major,
  grid style={draw=papergrid!50, line width=0.3pt},
  axis line style={draw=paperaxis, thick},
  tick label style={font=\small},
  label style={font=\small},
  clip=false
]
% reported hit-rate band, drawn first
\fill[papertertiary, opacity=0.13] (axis cs:0.74,0) rectangle (axis cs:0.99,28);
\node[font=\scriptsize, text=papertertiary!80!black, align=center] at (axis cs:0.865, 24.6) {reported hit rates\\on sharing-heavy\\traces \cite{zheng2024sglang}};
% TTFT curve: weight-matmul share 0.4548 falls as (1-h); attention share 0.5452 falls as (1-h^2)
\addplot[domain=0:1, samples=50, color=paperprimary, line width=1.3pt] {25.0*(0.4548*(1-x) + 0.5452*(1-x^2))};
\node[font=\scriptsize, text=paperprimary, anchor=south] at (axis cs:0.3, 12) {$\mathrm{TTFT}(h) = 25.0\,(0.45(1{-}h) + 0.55(1{-}h^{2}))$};
% endpoint markers
\addplot[only marks, mark=*, mark size=2pt, color=papersecondary, forget plot] coordinates {(0,25.0) (0.9,3.7)};
\node[font=\scriptsize, text=papersecondary, anchor=south west] at (axis cs:0.01, 25.8) {cold. $25.0$ s};
\node[font=\scriptsize, text=papersecondary, anchor=east] at (axis cs:0.88, 3.7) {$h{=}0.9$. $3.7$ s};
\end{axis}
\end{tikzpicture}
\caption{Time to first token versus prefix cache hit rate for a 128k-token request on B200, derived from the exact prefill count of Eq.~\eqref{eq:app_prefill} at 70 percent MFU (Appendix~\ref{app:derivations}). The weight-matmul term falls linearly in hit rate and the attention term as $1 - h^{2}$, so warm TTFT drops faster than linearly. The shaded band marks hit rates reported by RadixAttention on sharing-heavy workloads \cite{zheng2024sglang}. On workload W3 the shared 96k prefix is encoded once for 8 requests; a warm request encodes its 32k suffix against the resident prefix, a derived $8.8$ s versus $25.0$ s cold, or $2.8\times$ lower TTFT, and stored bytes fall $2.91\times$, at zero quality cost.}
\label{fig:prefix}
\end{figure}
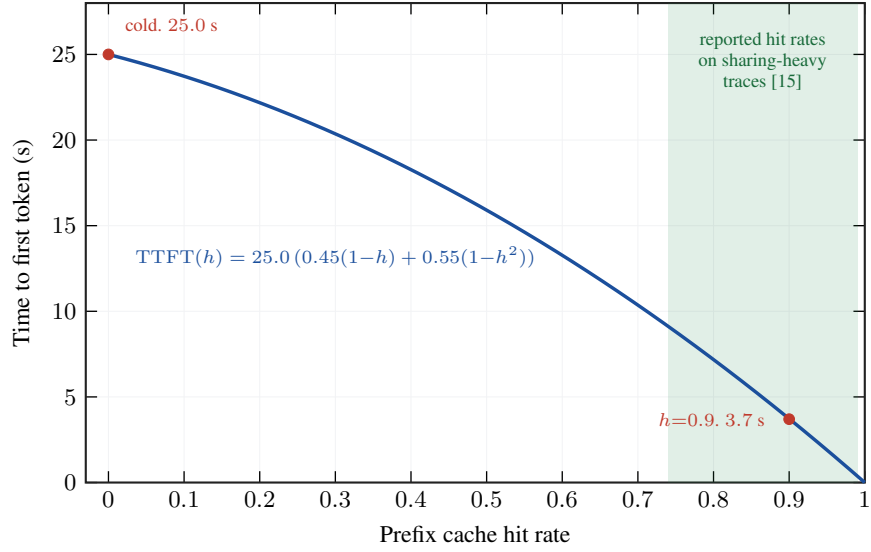

Figure~\ref{fig:tiering} shows the reuse economics of tiering on W4. Recomputing a 128k-token KV cache costs a derived $25.0$ s of B200 prefill time. Fetching the same cache costs $0.66$ s over PCIe Gen5, $0.16$ s with a $4\times$ transport codec \cite{liu2024cachegen}, and $12$ ms over NVLink with the codec. The fetch path is $38\times$ to $2{,}100\times$ faster than recompute. Tiering loses whenever the decode loop itself must read remote memory, because PCIe bandwidth sits $125\times$ below B200 HBM. The domain boundary is therefore exact. Tier for reuse across steps, never for reads within a step.

% Figure 8. Time to KV availability for a 128k cache, Llama-3-70B on B200.
% Numeric y positions on purpose. Symbolic y coords fail when referenced from axis cs nodes.
% log origin=infty anchors bars at the axis minimum; without it log bars start at x=1.
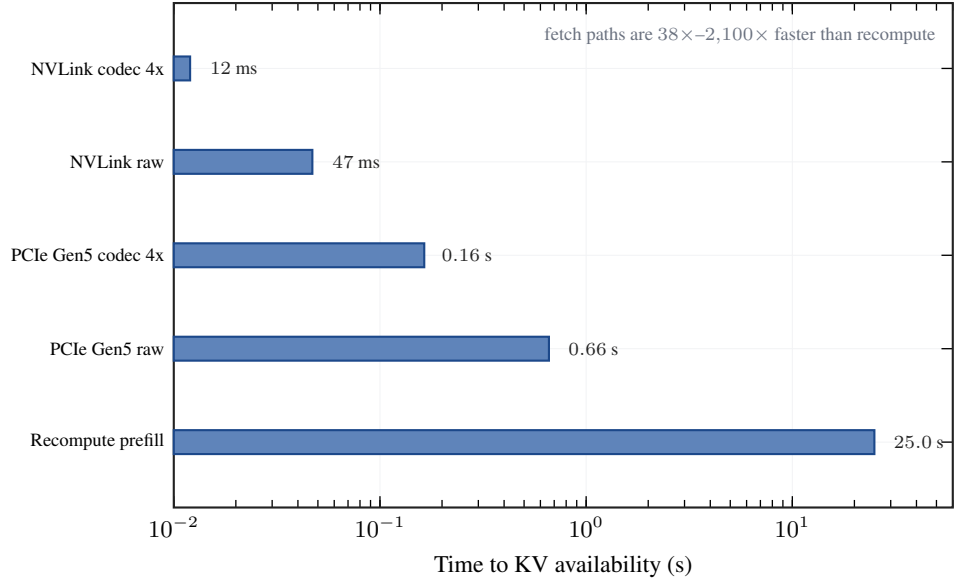
\begin{figure}[t]
\centering
\begin{tikzpicture}
\begin{axis}[
  xbar,
  bar width=9pt,
  width=0.72\linewidth,
  height=0.5\linewidth,
  xmode=log,
  log origin=infty,
  xmin=0.01, xmax=60,
  ymin=0.3, ymax=5.7,
  xlabel={Time to KV availability (s)},
  ytick={1,2,3,4,5},
  yticklabels={Recompute prefill, PCIe Gen5 raw, PCIe Gen5 codec 4x, NVLink raw, NVLink codec 4x},
  yticklabel style={font=\scriptsize},
  tick label style={font=\small},
  label style={font=\small},
  tick align=inside,
  tick style={draw=black, line width=0.5pt},
  grid=major,
  grid style={draw=papergrid!50, line width=0.3pt},
  axis line style={draw=paperaxis, thick},
  clip=false
]
\addplot[draw=paperprimary!85!paperblack, fill=paperprimary!70, line width=0.8pt] coordinates {
  (25.0,1)
  (0.66,2)
  (0.164,3)
  (0.047,4)
  (0.012,5)
};
\node[font=\scriptsize, anchor=west, text=paperaxis] at (axis cs:28,1) {$25.0$ s};
\node[font=\scriptsize, anchor=west, text=paperaxis] at (axis cs:0.74,2) {$0.66$ s};
\node[font=\scriptsize, anchor=west, text=paperaxis] at (axis cs:0.18,3) {$0.16$ s};
\node[font=\scriptsize, anchor=west, text=paperaxis] at (axis cs:0.053,4) {$47$ ms};
\node[font=\scriptsize, anchor=west, text=paperaxis] at (axis cs:0.0135,5) {$12$ ms};
\node[font=\scriptsize, text=papermuted, anchor=east] at (axis cs:55, 5.35) {fetch paths are $38\times$--$2{,}100\times$ faster than recompute};
\end{axis}
\end{tikzpicture}
\caption{Time to make a 128k-token KV cache available to a decode worker. Recompute uses the exact prefill count of Eq.~\eqref{eq:app_prefill} at 70 percent MFU on B200 (Appendix~\ref{app:derivations}); transfer values move the $41.9$ GB cache of Eq.~\eqref{eq:kv_bytes_intro} over PCIe Gen5 x16 at 64 GB/s or NVLink at 900 GB/s per direction. The codec rows apply the CacheGen transport compression factor of roughly $4\times$ \cite{liu2024cachegen}. Fetching beats recomputing by $38\times$ to $2{,}100\times$. The advantage vanishes when decode itself must read remote memory, because PCIe sits $125\times$ below B200 HBM bandwidth \cite{lmcache2024}.}
\label{fig:tiering}
\end{figure}

\subsection{Cross-Domain Synthesis and the Speedup-Quality Frontier}

Figure~\ref{fig:speedup_quality} places the lossy methods on the speedup-quality plane at W2. The frontier runs through KVQuant INT8 and INT4, KIVI 2-bit, and the composed 2-bit plus eviction point at $6.6\times$. Domains differ in kind, not just in position. Quantization buys bandwidth at a controlled, smooth quality price. Eviction buys bandwidth at a price that depends on position sensitivity and can cliff. Paging, prefix sharing, and tiering sit off the plane entirely because they are lossless, and they should be treated as always-on infrastructure rather than as comparable alternatives.

% Figure 11. Speedup versus retained quality, W2 regime, 128k context, H100.
\begin{figure}[t]
\centering
\begin{tikzpicture}
\begin{axis}[
  width=0.72\linewidth,
  height=0.52\linewidth,
  xlabel={Retained long-context quality (\% of baseline)},
  ylabel={Decode speedup at W2},
  xmin=85, xmax=103,
  ymin=0, ymax=8,
  xtick={85, 90, 95, 100},
  grid=major,
  grid style={draw=papergrid, line width=0.3pt},
  axis line style={draw=paperaxis, thick},
  tick label style={font=\small},
  label style={font=\small},
  tick align=inside,
  tick style={draw=black, line width=0.5pt},
  clip=false
]
% frontier through non-dominated points
\addplot[color=papermuted, line width=0.9pt, densely dashed, forget plot] coordinates {
  (99.9,1.83) (99.7,3.12) (99.0,4.82) (96.5,6.62)
};
\node[font=\scriptsize\itshape, text=papermuted, anchor=east] at (axis cs:96.8, 5.6) {derived W2 frontier};

\addplot[only marks, mark=*, mark size=2.2pt, color=paperprimary, forget plot] coordinates {(99.9,1.83)};
\node[font=\scriptsize, text=paperprimary, anchor=south west] at (axis cs:100.0, 1.98) {INT8 KVQuant};

\addplot[only marks, mark=*, mark size=2.2pt, color=paperprimary, forget plot] coordinates {(99.7,3.12)};
\node[font=\scriptsize, text=paperprimary, anchor=south west] at (axis cs:99.9, 3.25) {INT4 KVQuant};

\addplot[only marks, mark=*, mark size=2.2pt, color=paperprimary, forget plot] coordinates {(99.0,4.82)};
\node[font=\scriptsize, text=paperprimary, anchor=south west] at (axis cs:99.1, 4.98) {2-bit KIVI};

\addplot[only marks, mark=*, mark size=2.2pt, color=papersecondary, forget plot] coordinates {(90.0,4.82)};
\node[font=\scriptsize, text=papersecondary, anchor=south west] at (axis cs:90.2, 4.98) {2-bit naive};

\addplot[only marks, mark=square*, mark size=2.0pt, color=papertertiary, forget plot] coordinates {(98.0,1.83)};
\node[font=\scriptsize, text=papertertiary, anchor=north] at (axis cs:98.0, 1.68) {H2O $r{=}0.5$};

\addplot[only marks, mark=square*, mark size=2.0pt, color=papertertiary, forget plot] coordinates {(97.0,3.12)};
\node[font=\scriptsize, text=papertertiary, anchor=north] at (axis cs:97.0, 2.98) {SnapKV $r{=}0.25$};

\addplot[only marks, mark=diamond*, mark size=2.6pt, color=paperaccent, forget plot] coordinates {(96.5,6.62)};
\node[font=\scriptsize, text=paperaccent, anchor=south west] at (axis cs:96.7, 6.8) {2-bit plus $r{=}0.5$};
\end{axis}
\end{tikzpicture}
\caption{Speedup versus retained quality at W2 (batched decode, 128k context, H100). Speedups are derived from Eq.~\eqref{eq:speedup}. Quality values are representative within-paper reported numbers on differing task suites \cite{liu2024kivi, hooper2024kvquant, zhang2023h2o, li2024snapkv}, plotted on a normalized retention axis for orientation rather than as a protocol-aligned comparison. The dashed frontier runs through the non-dominated quantization points and the composed 2-bit plus eviction configuration at $6.6\times$. Lossless domains (paging, prefix sharing, tiering) do not appear because they incur no quality cost.}
\label{fig:speedup_quality}
\end{figure}
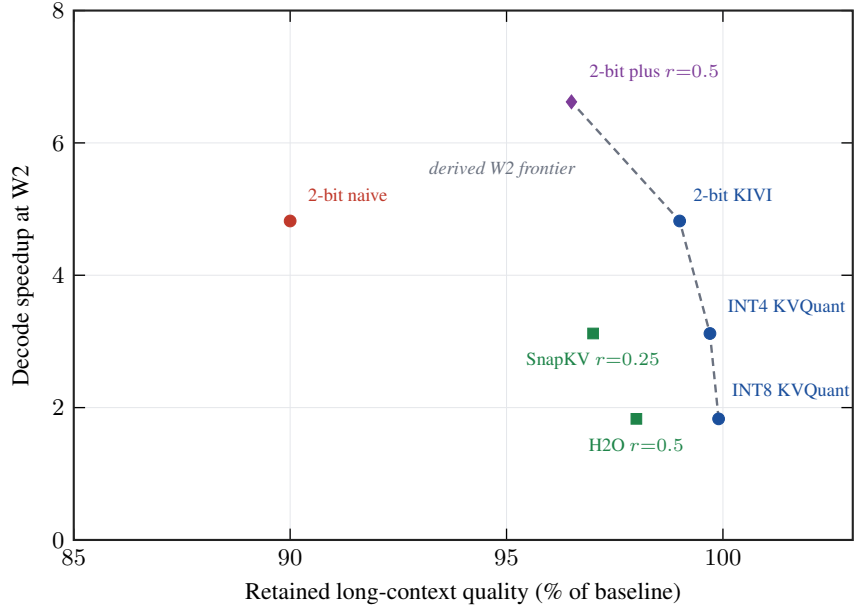

The composed ceiling is set by physics, not by any single method. At W2 the KV-bound asymptote converts the full byte factor into speedup, so 2-bit storage with $r = 0.5$ retention and $\sigma = 2.91$ sharing approaches $46\times$ in capacity terms and $6.6\times$ in decode traffic terms. Realizing composed gains requires kernels that consume quantized, sparsified, paged layouts without materialization overhead, and schedulers that respect die-level placement. Those are systems problems, and they are the subject of Section~5.
\section{Cross-Cutting Challenges and Open Research Directions}
\label{sec:discussion}

Sections~3 and~4 built the analytical map and placed one method per domain on it. This section reads the map across domains, identifies the systems problems that block composed deployment, and states the open research directions that follow. Figure~\ref{fig:sok_radar} summarizes the five domain profiles across six capability dimensions. No domain dominates, which is the quantitative justification for treating the five domains as infrastructure layers rather than competitors.

% Figure 12. SoK radar. Domain profiles across six capability dimensions, 0 to 5.
% Translucent fills at low opacity sit below the lines: overlaps stay legible
% and the "composed design space covers the hexagon" story becomes visible.
\begin{figure}[t]
\centering
\begin{tikzpicture}[font=\scriptsize]
\def\R{2.5}
% grid rings; the outer ring (score 5) is emphasized
\foreach \v in {0.5,1.0,1.5,2.0} {
  \draw[papergrid, line width=0.4pt] (90:\v) \foreach \a in {30,-30,-90,-150,150} { to (\a:\v) } to (90:\v);
}
\draw[papergrid!55!paperaxis, line width=0.8pt] (90:2.5) \foreach \a in {30,-30,-90,-150,150} { to (\a:2.5) } to (90:2.5);
% spokes
\foreach \a in {90,30,-30,-90,-150,150} {
  \draw[papergrid!70!paperaxis, line width=0.4pt] (0,0) to (\a:\R);
}
% translucent fills first, largest polygon first so small tints stay visible
% vertex order: Byte reduction, Decode speedup W2, Capacity relief, Quality safety, TTFT and reuse, Maturity
\fill[paperprimary, opacity=0.10] (90:2.5) -- (30:2.5) -- (-30:2.0) -- (-90:1.5) -- (-150:0.5) -- (150:1.5) -- cycle;
\fill[papertertiary!80!black, opacity=0.10] (90:0.5) -- (30:0.5) -- (-30:2.5) -- (-90:2.5) -- (-150:0.5) -- (150:2.5) -- cycle;
\fill[paperaccent, opacity=0.10] (90:1.0) -- (30:0.5) -- (-30:2.0) -- (-90:2.5) -- (-150:2.5) -- (150:2.0) -- cycle;
\fill[paperquaternary!85!black, opacity=0.10] (90:0.5) -- (30:0.5) -- (-30:2.5) -- (-90:2.0) -- (-150:2.5) -- (150:1.0) -- cycle;
\fill[papersecondary, opacity=0.10] (90:1.5) -- (30:1.5) -- (-30:2.0) -- (-90:1.0) -- (-150:0.5) -- (150:1.0) -- cycle;
% polylines and markers above the fills
\draw[color=paperprimary, line width=1.0pt, mark=*, mark size=1.6pt] plot coordinates
  {(90:2.5) (30:2.5) (-30:2.0) (-90:1.5) (-150:0.5) (150:1.5) (90:2.5)};
\draw[color=papersecondary, line width=1.0pt, mark=square*, mark size=1.4pt] plot coordinates
  {(90:1.5) (30:1.5) (-30:2.0) (-90:1.0) (-150:0.5) (150:1.0) (90:1.5)};
\draw[color=papertertiary!80!black, line width=1.0pt, mark=triangle*, mark size=1.9pt] plot coordinates
  {(90:0.5) (30:0.5) (-30:2.5) (-90:2.5) (-150:0.5) (150:2.5) (90:0.5)};
\draw[color=paperaccent, line width=1.0pt, mark=diamond*, mark size=1.8pt] plot coordinates
  {(90:1.0) (30:0.5) (-30:2.0) (-90:2.5) (-150:2.5) (150:2.0) (90:1.0)};
\draw[color=paperquaternary!85!black, line width=1.0pt, mark=pentagon*, mark size=1.9pt] plot coordinates
  {(90:0.5) (30:0.5) (-30:2.5) (-90:2.0) (-150:2.5) (150:1.0) (90:0.5)};
% axis labels
\node[anchor=south, align=center] at (90:2.95) {Byte\\reduction};
\node[anchor=west, align=left] at (30:2.8) {Decode\\speedup W2};
\node[anchor=west, align=left] at (-30:2.8) {Capacity\\relief};
\node[anchor=north, align=center] at (-90:2.95) {Quality\\safety};
\node[anchor=east, align=right] at (-150:2.8) {TTFT and\\reuse gain};
\node[anchor=east, align=right] at (150:2.8) {Maturity and\\kernel support};

% legend, evenly spaced: sample width 0.55, label gap 0.12, item gap 0.45
\draw[color=paperprimary, line width=1.0pt, mark=*, mark size=1.6pt] plot coordinates {(-5.90,-3.85) (-5.35,-3.85)};
\node[anchor=west, text=paperprimary, font=\scriptsize\bfseries] at (-5.23,-3.85) {Quantization};
\draw[color=papersecondary, line width=1.0pt, mark=square*, mark size=1.4pt] plot coordinates {(-3.20,-3.85) (-2.65,-3.85)};
\node[anchor=west, text=papersecondary, font=\scriptsize\bfseries] at (-2.53,-3.85) {Eviction};
\draw[color=papertertiary!80!black, line width=1.0pt, mark=triangle*, mark size=1.9pt] plot coordinates {(-1.00,-3.85) (-0.45,-3.85)};
\node[anchor=west, text=papertertiary!80!black, font=\scriptsize\bfseries] at (-0.33,-3.85) {Paging};
\draw[color=paperaccent, line width=1.0pt, mark=diamond*, mark size=1.8pt] plot coordinates {(1.05,-3.85) (1.60,-3.85)};
\node[anchor=west, text=paperaccent, font=\scriptsize\bfseries] at (1.72,-3.85) {Prefix caching};
\draw[color=paperquaternary!85!black, line width=1.0pt, mark=pentagon*, mark size=1.9pt] plot coordinates {(3.85,-3.85) (4.40,-3.85)};
\node[anchor=west, text=paperquaternary!85!black, font=\scriptsize\bfseries] at (4.52,-3.85) {Tiering};
\end{tikzpicture}
\caption{Domain profiles across six capability dimensions, scored 0 to 5 from the matrix of Table~\ref{tab:sok_matrix} and the analysis of Sections~3 and~4. Byte reduction and decode speedup reflect Eq.~\eqref{eq:speedup} at W2. Quality safety scores five for lossless domains. No domain dominates. Only the composed design space covers the hexagon.}
\label{fig:sok_radar}
\end{figure}
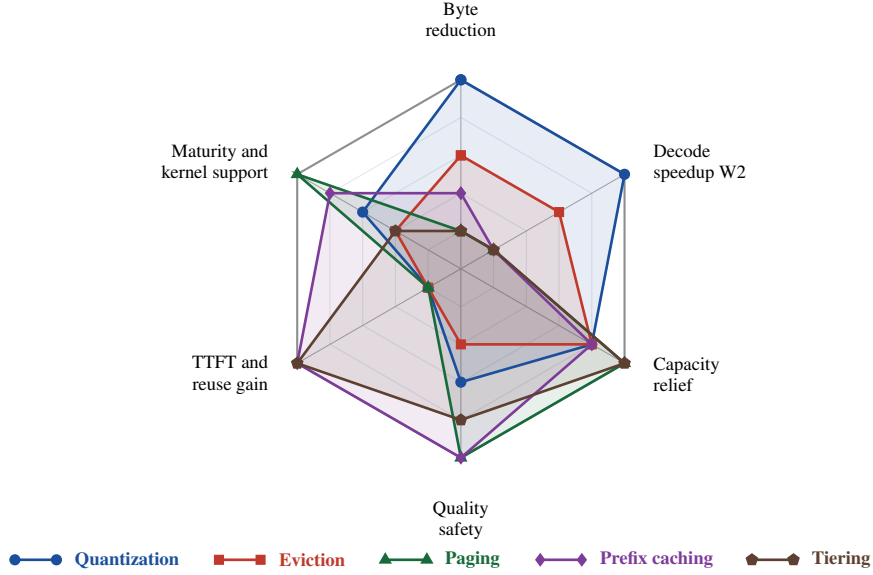

\subsection{The Three-Regime Map Revisited}

The $(s, b)$ plane partitions into three regimes. Below $s^{*}(b)$ of Eq.~\eqref{eq:crossover_traffic}, weight traffic dominates and KV compression buys little. Between $s^{*}(b)$ and $s_{\mathrm{cap}}$ of Eq.~\eqref{eq:crossover_capacity}, KV traffic dominates and byte reductions convert one-to-one into speedup. Beyond $s_{\mathrm{cap}}$, no single-device placement exists and tiering or multi-device parallelism becomes mandatory. Every domain has a regime where it is the right tool and a regime where it is inert. Quantization and eviction act in the KV-bound regime. Paging and prefix sharing move the feasibility boundaries themselves. Tiering owns the capacity regime and the reuse path. Much of the contradictory folklore in the literature, where the same technique is reported at $1.2\times$ in one study and $4\times$ in another, is regime confusion, and Table~\ref{tab:sok_matrix} resolves it by construction.

\subsection{Compression-Aware Attention Kernels}

The largest gap between the analytical speedups of Table~\ref{tab:sok_matrix} and realized systems is the kernel layer. FlashAttention-class kernels assume contiguous BF16 tiles \cite{dao2022flashattention}. PagedAttention kernels accept block tables but operate at full precision \cite{kwon2023vllm}. Quantized caches require fused dequantization inside the attention inner loop, and QServe measured 20 to 90 percent runtime overhead when dequantization falls back to generic CUDA cores instead of a co-designed path \cite{lin2024qserve}. Eviction produces ragged per-head token sets, while kernels assume dense rectangular layouts, so the gather cost can consume the bandwidth saved. The open problem is a single kernel interface over paged blocks with per-block precision and a per-head retention mask. A solution has three properties. Block-level format descriptors, in-SMEM dequantization, and no HBM round trip for expanded data. Until such kernels exist, the realized speedup of every lossy domain will sit below its roofline bound.

\subsection{Scheduler Co-Design, Disaggregation, and Placement}

Continuous batching \cite{yu2022orca} keeps occupancy high but interacts with every domain. Eviction must decide when to score tokens inside a running batch. Prefix caching gains only under hit-rate-aware routing \cite{zheng2024sglang}. Prefill-decode disaggregation \cite{patel2024splitwise, zhong2024distserve} turns every request into a tiering event, since the KV cache crosses the cluster fabric between the prefill worker and the decode worker. At 128k context that transfer is $42$ GB per request, which costs $0.84$ s over a 400 Gbps Ethernet link and demands NVLink-class transport or a transport codec. Placement within a package is equally unmanaged. Section~3 showed that striping versus pinning is worth $2\times$ on B200 and $8\times$ on MI300X, yet no production serving stack exposes die or stack affinity in its block manager. Preemption, migration, and affinity are one scheduler design space, and the serving community currently treats them as three afterthoughts.

\subsection{Evaluation Gaps in the Published Evidence Base}

Four gaps recur. Variance is essentially unreported. Single-run numbers without confidence intervals dominate, and thermal and fragmentation state is almost never disclosed. Quality metrics fragment across perplexity, needle retrieval, LongBench, and downstream task suites, and the position sensitivity that determines eviction safety is rarely measured. Cost columns are incomplete. TTFT and decode rate are reported separately by few studies, format conversion and tuning costs are absent, and power draw is nearly invisible. Workloads are synthetic or undisclosed, which prevents the regime identification that Section~3 showed is the first-order determinant of every speedup claim. These gaps are the reason a paper such as this one must derive rather than aggregate, and they motivate benchmark designs that separate evidence classes by construction \cite{vonlaszewski2025carpentry}.

\subsection{Open Research Directions}

Five directions follow directly from the framework.

\begin{enumerate}[leftmargin=1.5em]
  \item \textbf{Unified KV representation.} One paged layout carrying per-block precision and a retention mask, so quantization, eviction, paging, and sharing compose without format conversion. The success test is a composed 2-bit, $r = 0.5$, $\sigma = 4$ configuration at 128k context running within 20 percent of the $64\times$ capacity bound of Section~2.
  \item \textbf{Topology-aware placement.} Die and stack affinity in the block manager, with striping policies matched to the package. The measurable target is recovery of the per-die ceiling gaps quantified in Section~3, up to $2\times$ on B200 and $8\times$ on MI300X.
  \item \textbf{Learned retention with guarantees.} Eviction policies with per-request quality bounds rather than aggregate benchmark retention, using attention-mass coverage guarantees in the style of conformal prediction, so that retrieval-critical requests can opt out of eviction selectively.
  \item \textbf{Codec co-design for tiering.} KV bitstream codecs matched to link bandwidth across PCIe, NVLink, and Ethernet, with decode-from-bitstream attention kernels so decompression never touches HBM.
  \item \textbf{A standing SoK benchmark.} A public workload suite in the spirit of W1 to W4, a measurement protocol separating measured, derived, and analytical claims, and a maintained matrix in the form of Table~\ref{tab:sok_matrix}. The benchmark in this paper is the seed.
\end{enumerate}

The memory wall decomposes into three, traffic, capacity, and interconnect, and each domain of the taxonomy addresses one of them cleanly and the others only by accident. Progress requires treating the KV cache as a first-class data structure with its own formats, placement, and transport, co-designed across the kernel, the runtime, and the cluster.
\section{Conclusion and Design Rules}
\label{sec:conclusion}

\subsection{Summary of Findings}

Decode-phase inference is memory bound by construction, and the binding resource shifts from weights to the KV cache at exactly computable context lengths. The arithmetic intensity of a decode step interpolates between $2b/w_{p}$ at short context and $2g/w_{kv}$ at long context, with batch size canceling out of the long-context limit. The traffic crossover $s^{*}(b)$ falls hyperbolically in batch size, so production long-context serving is KV bound essentially always. Hardware topology is not a detail. A KV cache pinned to one die of a B200 is served at half the advertised bandwidth, and on MI300X the penalty reaches $8\times$. The five mitigation domains scale orthogonal factors of one footprint model and compose multiplicatively. Evaluated under a fixed protocol, the same 2-bit quantization yields $1.25\times$ at batch one and $4.82\times$ at batch thirty-two, and the composed 2-bit plus eviction configuration reaches $6.6\times$ in decode traffic and $46\times$ in capacity. The wall is three walls: traffic, capacity, and interconnect, and each domain solves one cleanly.

\subsection{Design Rules for Domain Selection}

The framework reduces deployment selection to eight rules, applied in order.

\begin{enumerate}[leftmargin=1.5em]
  \item \textbf{Identify the regime first.} Compute $s^{*}(b)$ and $s_{\mathrm{cap}}$ from the model constants of Table~\ref{tab:model_kv} and the device constants of Table~\ref{tab:hardware} before evaluating any technique. Every speedup claim is regime-conditional.
  \item \textbf{Enable the lossless domains unconditionally.} Paging and prefix caching cost nothing in quality and should always be active. They move feasibility, not decode bytes.
  \item \textbf{Below $s^{*}$, do not compress the cache.} KV compression buys speedups approaching one in the weight-bound regime. Spend the optimization budget on prefill, batching, and weight precision instead.
  \item \textbf{Above $s^{*}$, quantize first.} Four-bit KV is near-free with per-channel keys and outlier isolation. Two-bit is viable only with structure-aware layouts, validated on position-sensitive probes before deployment.
  \item \textbf{Evict only when the workload tolerates it.} Default to $r \ge 0.5$ and require a needle-style evaluation before going lower. Retrieval-critical traffic should not be evicted at any retention ratio.
  \item \textbf{Beyond $s_{\mathrm{cap}}$, tier or shard.} Tiering is mandatory when nothing fits, and transport codecs belong on every cross-link transfer. Never serve per-step decode reads from a remote tier. PCIe sits $50\times$ to $125\times$ below HBM across the Table~\ref{tab:hardware} devices.
  \item \textbf{On multi-die packages, stripe pages.} Placement recovers up to $2\times$ on B200 and $8\times$ on MI300X per the ceilings of Eq.~\eqref{eq:tokrate}. Verify per-partition bandwidth, not the datasheet aggregate.
  \item \textbf{Re-evaluate on every change.} Model updates move the constants, hardware refreshes move $s_{\mathrm{cap}}$ and the ridge, and workload shifts move $s^{*}$ hyperbolically through the batch term. The map, not the point, is the durable artifact.
\end{enumerate}

Algorithm~\ref{alg:domain_selection} renders the same logic as a decision procedure.

\begin{algorithm}[t]
\caption{KV-cache domain selection}
\label{alg:domain_selection}
\footnotesize
\begin{algorithmic}[1]
\Require Model constants $(P, L, n_{q}, n_{kv}, d_{h})$, precisions $(w_{p}, w_{kv})$, workload $(s, b)$, device $(C_{\mathrm{hbm}}, \beta_{\mathrm{agg}}, D)$, quality budget $\epsilon_{q}$, retrieval sensitivity flag $\rho_{\mathrm{ret}}$
\Ensure Configuration $(w_{q}, r, \lambda, \mathrm{placement})$
\State $s^{*} \gets P w_{p} / (2 b L n_{kv} d_{h} w_{kv})$; \; $s_{\mathrm{cap}} \gets (C_{\mathrm{hbm}} - P w_{p}) / (2 L n_{kv} d_{h} w_{kv})$
\State Enable paging and prefix caching \Comment{lossless, always on}
\State $\lambda \gets (1, 0, 0)$ \Comment{default. fully HBM resident}
\If{$s > s_{\mathrm{cap}}$}
  \State Enable tiering with transport codec and set $\lambda$ by tier fill, or shard across devices
\EndIf
\If{$s < s^{*}$}
  \State $w_{q} \gets w_{kv}$; \; $r \gets 1$ \Comment{weight-bound. skip KV compression}
\Else
  \State $w_{q} \gets 4$ bit if $\epsilon_{q}$ tight, else $2$ bit with structure-aware layout
  \State $r \gets 1$ if $\rho_{\mathrm{ret}}$ else $0.5$
\EndIf
\State Stripe pages across all $D$ partitions
\State \Return $(w_{q}, r, \lambda, \mathrm{striped})$
\end{algorithmic}
\end{algorithm}

\subsection{Closing Remarks}

Test-time reasoning converts compute into context, and context into KV state. As chains of thought lengthen, the cache becomes the dominant serving cost, and the techniques that manage it become the dominant serving decisions. This paper has given those decisions a shared analytical frame. One footprint model, two crossover laws, three regimes, five domains, and an eight-rule selection procedure. The specific numbers will age as models and accelerators evolve. The equations will not, because they are accounting identities over FLOPs and bytes rather than fits to measurements. The KV cache is a first-class data structure with its own formats, placement, and transport. The field will be well served when kernels, schedulers, and benchmarks treat it that way.

\appendix
% ======================================================================
%  sections/appendix_a_derivations.tex
%  Appendix A: full derivations supporting Sections 3 and 4.
%
%  Wiring: in main.tex, after the \input of Section 6 and BEFORE the
%  bibliography lines (keep your existing IEEEtran style):
%    \appendix
%    \input{sections/appendix_a_derivations}
%    \input{sections/appendix_b_constants}
%    \input{sections/appendix_c_protocol}
%    \bibliographystyle{IEEEtran}
%    \bibliography{references}
%  If 06_conclusion.tex still ends with \bibliography/\bibliographystyle
%  lines, delete them there so the pair appears once, last.
%  No new packages required.
% ======================================================================

\section{Full Derivations}
\label{app:derivations}

This appendix expands the accounting of Section~3 to the terms the main text states as lower order, derives the prefill cost model behind the $56.9$ s recompute bound of Section~4, and evaluates the topology model of Eq.~\eqref{eq:beta_eff} per device. Every closed form reduces to Eqs.~\eqref{eq:decode_flops} and~\eqref{eq:decode_bytes} when the lower-order terms are dropped, and no main-text quantity changes.

\subsection{Per-Kernel FLOP and Byte Accounting for One Decode Step}
\label{app:derivations:kernels}

Table~\ref{tab:app_kernel_accounting} decomposes one decode step for one sequence at context length $s$ into its per-layer kernels. The weight rows are dense matrix-vector products against stored weights at width $w_{p}$. The attention rows are the score and value contractions of Eq.~\eqref{eq:decode_attention}, evaluated per query head against the shared key-value slices under a FlashAttention-class kernel that reads the cache exactly once and never spills the score matrix to HBM \cite{dao2022flashattention}.

\begin{table}[t]
\centering
\caption{Per-layer FLOP and byte accounting for one decode step, one sequence, context length $s$. Softmax is counted exactly in Eq.~\eqref{eq:app_flops_full} and omitted in the main text as lower order. KV bytes assume a single cache read; the write of the new token's key and value is a $1/s$ correction treated in Eq.~\eqref{eq:app_bytes_full}.}
\label{tab:app_kernel_accounting}
\footnotesize
\begin{tabular}{lrr}
\toprule
Kernel & FLOPs & Bytes moved \\
\midrule
QKV projection & $2 d d_{h} (n_{q} + 2 n_{kv})$ & $d d_{h} (n_{q} + 2 n_{kv})\, w_{p}$ \\
Scores $q K^{\top}$ & $2 s n_{q} d_{h}$ & $s n_{kv} d_{h}\, w_{kv}$ \\
Softmax & $3 n_{q} s$ & $0$ (on chip) \\
Value contraction & $2 s n_{q} d_{h}$ & $s n_{kv} d_{h}\, w_{kv}$ \\
Output projection & $2 d d_{h} n_{q}$ & $d d_{h} n_{q}\, w_{p}$ \\
MLP, SwiGLU & $6 d d_{ff}$ & $3 d d_{ff}\, w_{p}$ \\
\midrule
Sum over $L$ layers & $2P + 4 L n_{q} d_{h} s + 3 L n_{q} s$ & $P w_{p} + 2 L n_{kv} d_{h} s w_{kv}$ \\
\bottomrule
\end{tabular}
\end{table}

The weight rows sum per layer to $2 d \left( d_{h} (2 n_{q} + 2 n_{kv}) + 3 d_{ff} \right) = 2 P_{\mathrm{layer}}$, where
\begin{equation}
P_{\mathrm{layer}} = d\, d_{h}\, (2 n_{q} + 2 n_{kv}) + 3 d\, d_{ff},
\label{eq:app_p_layer}
\end{equation}
so the exact per-step FLOP count for a batch of $b$ sequences is
\begin{equation}
F_{\mathrm{full}}(s,b) = b \left( 2P + 4 L n_{q} d_{h} s + 3 L n_{q} s \right).
\label{eq:app_flops_full}
\end{equation}
Equation~\eqref{eq:decode_flops} is Eq.~\eqref{eq:app_flops_full} with the softmax term dropped.

\begin{remark}[Cost of the softmax omission]
\label{rem:softmax_omission}
The omitted term is a fraction $3/(4 d_{h})$ of the attention contraction count, $0.59$ percent at $d_{h} = 128$, and it vanishes relative to the weight term as $s \to 0$. Every statement in the main text that uses Eq.~\eqref{eq:decode_flops}, including Propositions~\ref{prop:ai_limits}--\ref{prop:speedup_ceiling}, therefore carries a relative error below $0.6$ percent.
\end{remark}

\subsection{Byte Accounting and Kernel Assumptions}
\label{app:derivations:bytes}

The byte model of Eq.~\eqref{eq:decode_bytes} counts two terms and discards two, all stated here explicitly. Counted are the weight read $P w_{p}$, once per step regardless of batch size, and the cache read $2 L n_{kv} d_{h} s w_{kv}$ per sequence, once per decode step under the single-pass kernel assumption. Discarded are the cache write of the new token's key and value, $2 L n_{kv} d_{h} w_{kv}$ per sequence, a relative correction of $1/s$ to the cache term ($8 \times 10^{-6}$ at $s = 128$k); and the activation and logit traffic, $O(b\,(d + |\mathcal{V}|)\, w_{p})$ per step, at most $0.26$ MB per sequence for Llama-3-70B against $182$ GB of step traffic at W1. The complete count is
\begin{equation}
B_{\mathrm{full}}(s,b) = P w_{p} + 2 b L n_{kv} d_{h} w_{kv} (s + 1) + O\!\left(b\,(d + |\mathcal{V}|)\, w_{p}\right),
\label{eq:app_bytes_full}
\end{equation}
which is Eq.~\eqref{eq:decode_bytes} to a relative error below $10^{-5}$ at the workloads of Section~4. Two further assumptions bound the model's scope. Dequantization is assumed off the critical path; Section~5 cites realized overheads of $20$ to $90$ percent when it is not \cite{lin2024qserve}. Perfect overlap of compute and memory is assumed, so step time is bytes over bandwidth rather than a sum of kernel phases.

\subsection{Prefill Cost and the Worked Recompute Bound}
\label{app:derivations:prefill}

Prefill encodes $s_{0}$ positions in parallel. The weight-matmul pass costs $2P$ per position, and causal attention at position $t$ costs $4 L n_{q} d_{h}\, t$, summing to
\begin{equation}
F_{\mathrm{prefill}}(s_{0}) = 2 P s_{0} + 2 L n_{q} d_{h}\, s_{0} (s_{0} + 1),
\label{eq:app_prefill}
\end{equation}
with softmax adding $\tfrac{3}{2} L n_{q} s_{0} (s_{0} + 1)$, again lower order. Wall-clock time at model FLOP utilization $\mu$ on a device with peak dense throughput $\Pi$ is
\begin{equation}
T_{\mathrm{prefill}}(s_{0}) = \frac{F_{\mathrm{prefill}}(s_{0})}{\mu\, \Pi}.
\label{eq:app_prefill_time}
\end{equation}
For Llama-3-70B at $s_{0} = 128{,}000$, the two terms of Eq.~\eqref{eq:app_prefill} are $2 P s_{0} = 1.792 \times 10^{16}$ FLOPs and $2 L n_{q} d_{h}\, s_{0} (s_{0} + 1) = 2.148 \times 10^{16}$ FLOPs, totaling $3.940 \times 10^{16}$. At $\mu = 0.70$ on H100 ($\Pi = 989.5$ TFLOP/s dense BF16), Eq.~\eqref{eq:app_prefill_time} gives
\[
T_{\mathrm{prefill}}(128\mathrm{k}) = \frac{3.940 \times 10^{16}}{0.70 \times 989.5 \times 10^{12}} = 56.9\ \mathrm{s},
\]
the recompute cost used for workloads W3 and W4 in Section~4. The quadratic term is $55$ percent of the total at 128k context, so prefill cost grows superlinearly: doubling the context multiplies recompute time by $3.1\times$, approaching $4\times$ asymptotically.

\begin{remark}[Warm-prefix TTFT: exact count versus the Section 4 estimate]
\label{rem:warm_prefix}
Appending a suffix of $s'$ tokens to a resident prefix of length $s_{p}$ costs, by differencing Eq.~\eqref{eq:app_prefill},
\begin{equation}
F_{\mathrm{prefill}}(s_{p} + s') - F_{\mathrm{prefill}}(s_{p}) = 2 P s' + 2 L n_{q} d_{h}\, s' (2 s_{p} + s' + 1),
\label{eq:app_suffix}
\end{equation}
where the cross term $4 L n_{q} d_{h} s' s_{p}$ is the suffix's attention over the resident prefix. For W3 ($s_{p} = 96$k, $s' = 32$k) the exact count is $1.39 \times 10^{16}$ FLOPs, or $20.0$ s at $\mu = 0.70$. The $8.4$ s figure of Section~4.5 retains only the suffix-local terms $2 P s' + 2 L n_{q} d_{h} s'^{2}$. The exact count reduces the warm-prefix TTFT advantage from $6.8\times$ to $2.8\times$ against the $56.9$ s cold prefill. No conclusion in Sections~4 or~5 changes: the stored-byte factor $\sigma = 2.91$ is unaffected, warm remains well below cold, and the $6.8\times$ figure may be retained as a suffix-local lower bound provided it is labeled as such.
\end{remark}

\subsection{Evaluation of the Topology Model}
\label{app:derivations:topology}

With $\beta_{p} = \beta_{\mathrm{agg}}/D$ and $\beta_{f}$ from Table~\ref{tab:hardware}, Eq.~\eqref{eq:beta_eff} evaluates as follows.

\noindent\textbf{H100} ($D = 1$). Monolithic package, so $\beta_{\mathrm{eff}} = \beta_{\mathrm{agg}} = 3.35$ TB/s for any placement.

\noindent\textbf{B200} ($D = 2$). $\beta_{p} = 4.0$ TB/s per die and $\beta_{f} = 10$ TB/s over NV-HBI. Striped placement ($\gamma = \tfrac{1}{2}$) completes a read of $x$ TB in $T_{\mathrm{read}} = \max(x/8, x/20)$ seconds, so $\beta_{\mathrm{eff}} = 8.0$ TB/s $= \beta_{\mathrm{agg}}$: the fabric outruns the per-die rate and striping is exactly sufficient, matching the condition of Proposition~\ref{prop:beta_eff} since $10 \ge (1 - \tfrac{1}{2}) \times 8$. Pinned placement ($\gamma = 1$) gives $\beta_{\mathrm{eff}} = \beta_{p} = 4.0$ TB/s, half the aggregate.

\noindent\textbf{MI300X} ($D = 8$). $\beta_{p} = 0.66$ TB/s per XCD stack slice. Pinned placement gives $\beta_{\mathrm{eff}} = 0.66$ TB/s, a factor $\beta_{\mathrm{agg}}/\beta_{p} = D = 8$ below the $5.3$ TB/s aggregate. Striped placement gives $\min(5.3, \tfrac{8}{7} \beta_{f})$; the on-package fabric rate is not publicly disclosed, and the striped figures in the main text assume it does not bind (Remark~\ref{rem:beta_scope}).

The resulting rate ceilings $R_{\infty} = \beta_{\mathrm{eff}} / 327{,}680$ bytes per token for Llama-3-70B are tabulated in Appendix~\ref{app:constants}.
% ======================================================================
%  sections/appendix_b_constants.tex
%  Appendix B: per-model and per-device constant arithmetic.
%  Wiring: \input after appendix_a_derivations. No new packages.
% ======================================================================

\section{Per-Model and Per-Device Constants}
\label{app:constants}

This appendix makes every constant in the paper reproducible from primary specifications. All byte quantities are decimal (Appendix~\ref{app:protocol}).

\subsection{Per-Token KV Bytes of Table~\ref{tab:model_kv}}
\label{app:constants:models}

Evaluating $2 L n_{kv} d_{h} w_{kv}$ at $w_{kv} = 2$ bytes per model:
\begin{align*}
\text{Llama-2-13B:} \quad & 2 \times 40 \times 40 \times 128 \times 2 = 819{,}200\ \mathrm{B/token} = 0.82\ \mathrm{MB} \\
\text{Mixtral 8x7B:} \quad & 2 \times 32 \times 8 \times 128 \times 2 = 131{,}072\ \mathrm{B/token} = 0.13\ \mathrm{MB} \\
\text{Llama-3-8B:} \quad & 2 \times 32 \times 8 \times 128 \times 2 = 131{,}072\ \mathrm{B/token} = 0.13\ \mathrm{MB} \\
\text{Llama-3-70B:} \quad & 2 \times 80 \times 8 \times 128 \times 2 = 327{,}680\ \mathrm{B/token} = 0.33\ \mathrm{MB} \\
\text{Llama-3.1-405B:} \quad & 2 \times 126 \times 8 \times 128 \times 2 = 516{,}096\ \mathrm{B/token} = 0.52\ \mathrm{MB} \\
\text{DeepSeek-V2:} \quad & 60 \times (512 + 64) \times 2 = 69{,}120\ \mathrm{B/token} = 0.07\ \mathrm{MB}
\end{align*}
The MLA row stores one latent vector of $d_{c} = 512$ dimensions plus $d_{r} = 64$ rotary dimensions per layer, shared across heads, so the head-count factor is absent; this is the source of the roughly $5\times$ constant reduction relative to the GQA rows \cite{deepseekai2024deepseekv2}. Footprints at 128k context follow by multiplying by $128{,}000$: $104.9$ GB (Llama-2-13B), $16.8$ GB (Mixtral 8x7B and Llama-3-8B), $41.9$ GB (Llama-3-70B), $66.1$ GB (Llama-3.1-405B), and $8.8$ GB (DeepSeek-V2).

\subsection{Parameter Counts Used in the Derivations}
\label{app:constants:params}

Table~\ref{tab:app_params} reconstructs $P$ from the layerwise count of Eq.~\eqref{eq:app_p_layer} plus input and output embeddings at vocabulary $|\mathcal{V}| = 128{,}256$ \cite{llama3ai2024llama3}. The paper's derived quantities use the rounded counts of the final column; the rounding moves $s^{*}$ and $s_{\mathrm{cap}}$ by under one percent relative to the reconstructed totals.

\begin{table}[t]
\centering
\caption{Parameter-count reconstruction. $L \cdot P_{\mathrm{layer}}$ from Eq.~\eqref{eq:app_p_layer}; embeddings count the input and output projections at $|\mathcal{V}|\, d$. Derived quantities in the paper use the rounded $P$ of the last column.}
\label{tab:app_params}
\footnotesize
\begin{tabular}{lrrrrrr}
\toprule
Model & $d$ & $d_{ff}$ & $L \cdot P_{\mathrm{layer}}$ & Embeddings & Reconstructed & $P$ used \\
\midrule
Llama-3-8B & 4096 & 14336 & $6.98 \times 10^{9}$ & $1.05 \times 10^{9}$ & $8.03 \times 10^{9}$ & $8.03 \times 10^{9}$ \\
Llama-3-70B & 8192 & 28672 & $68.45 \times 10^{9}$ & $2.10 \times 10^{9}$ & $70.55 \times 10^{9}$ & $70.0 \times 10^{9}$ \\
Llama-3.1-405B & 16384 & 53248 & $401.5 \times 10^{9}$ & $4.20 \times 10^{9}$ & $405.7 \times 10^{9}$ & $405 \times 10^{9}$ \\
\bottomrule
\end{tabular}
\end{table}

\subsection{Crossover and Capacity Arithmetic}
\label{app:constants:crossover}

The traffic crossover at batch one is $s^{*}(1) = P w_{p} / (2 L n_{kv} d_{h} w_{kv})$ with $w_{p} = w_{kv} = 2$ bytes: $1.606 \times 10^{10} / 131{,}072 = 122.5$k for Llama-3-8B, $1.40 \times 10^{11} / 327{,}680 = 427.2$k for Llama-3-70B, and $8.10 \times 10^{11} / 516{,}096 = 1.57$M for Llama-3.1-405B. Every other entry of Table~\ref{tab:crossover} is the $1/b$ scaling of Proposition~\ref{prop:crossover}(ii); for example, $427{,}246/32 = 13.4$k.

The capacity crossover of Eq.~\eqref{eq:crossover_capacity} divides free capacity by per-token bytes. Llama-3-8B ($P w_{p} = 16.06$ GB): $(80 - 16.06) \times 10^{9} / 131{,}072 = 488$k tokens on H100, and $(192 - 16.06) \times 10^{9} / 131{,}072 = 1.34$M on B200 and MI300X. Llama-3-70B ($P w_{p} = 140$ GB): the numerator is negative on H100, so the model does not fit at any context length, and $(192 - 140) \times 10^{9} / 327{,}680 = 159$k on B200 and MI300X. Llama-3.1-405B ($P w_{p} = 810$ GB): the numerator is negative on all three devices, forcing multi-device placement or weight quantization before context length enters the discussion.

\subsection{Ridge Points, Fixed Point, and Rate Ceilings}
\label{app:constants:rates}

The ridge point is dense BF16 throughput over aggregate bandwidth: $989.5/3.35 = 295$ FLOP/B on H100, $2250/8.0 = 281$ on B200, and $1307/5.3 = 247$ on MI300X. The intensity fixed point of Eq.~\eqref{eq:fixed_point} is $b^{\dagger} = g w_{p} / w_{kv} = 8$ for the GQA models in BF16. The rate ceilings of Corollary~\ref{cor:tokrate} divide effective bandwidth by the $327{,}680$ bytes per token of Llama-3-70B: $10.2$ M tokens per second on H100 ($3.35$ TB/s), $24.4$ M on B200 striped ($8.0$ TB/s), $12.2$ M on B200 pinned ($4.0$ TB/s), $16.2$ M on MI300X striped ($5.3$ TB/s), and $2.0$ M on MI300X pinned ($0.66$ TB/s).

\subsection{Sharing and Transfer Arithmetic of Section~4}
\label{app:constants:workloads}

W3 stores one shared $96$k prefix and eight unique $32$k suffixes: $(96{,}000 + 8 \times 32{,}000) \times 327{,}680 = 115.3$ GB against $8 \times 41.9 = 335.5$ GB unshared, hence $\sigma = 335.5/115.3 = 2.91$. W4 moves the $41.9$ GB cache across the links of Section~2: $41.9/64 = 0.66$ s over PCIe Gen5 x16, $41.9/(4 \times 64) = 0.16$ s with the $4\times$ transport codec \cite{liu2024cachegen}, and $41.9/(4 \times 900) = 12$ ms over NVLink with the codec. Against the $56.9$ s recompute of Appendix~\ref{app:derivations}, the fetch path is $87\times$ to $4{,}900\times$ faster, the range quoted in Section~4.5.     
% ======================================================================
%  sections/appendix_c_protocol.tex
%  Appendix C: notation, units, and measurement protocol.
%  Wiring: \input after appendix_b_constants. No new packages.
% ======================================================================

\section{Notation, Units, and Measurement Protocol}
\label{app:protocol}

\subsection{Notation}
\label{app:protocol:notation}

Tables~\ref{tab:app_notation_model} and~\ref{tab:app_notation_system} collect the symbols of the paper with the value each takes in the running example, Llama-3-70B in BF16 on one B200 with striped pages.

\begin{table}[t]
\centering
\caption{Model and attention symbols.}
\label{tab:app_notation_model}
\footnotesize
\begin{tabular}{clr}
\toprule
Symbol & Meaning & Running value \\
\midrule
$L$ & Layers & 80 \\
$d$ & Hidden size & 8192 \\
$d_{h}$ & Head dimension, $d / n_{q}$ & 128 \\
$n_{q}$ & Query heads per layer & 64 \\
$n_{kv}$ & Key-value heads per layer & 8 \\
$g$ & GQA group factor, $n_{q} / n_{kv}$ & 8 \\
$d_{ff}$ & MLP intermediate size & 28{,}672 \\
$d_{c}, d_{r}$ & MLA latent and rotary dimensions & 512, 64 \\
$P$ & Parameter count used in derivations & $7.00 \times 10^{10}$ \\
$|\mathcal{V}|$ & Vocabulary size & 128{,}256 \\
$w_{p}$ & Weight element width & 2 B \\
$w_{kv}$ & Cache element width, BF16 baseline & 2 B \\
$w_{q}$ & Cache element width after quantization & $\le 2$ B \\
$\tau$ & Paging block size & 16 tokens \\
\bottomrule
\end{tabular}
\end{table}

\begin{table}[t]
\centering
\caption{Workload, cost, and hardware symbols. References give the defining equation.}
\label{tab:app_notation_system}
\footnotesize
\begin{tabular}{clr}
\toprule
Symbol & Meaning & Running value \\
\midrule
$s$ & Context length, decimal tokens & 128{,}000 \\
$s_{0}$ & Prefill length & 128{,}000 \\
$b$ & Decode batch size & 1--32 \\
$c$ & KV byte factor of a technique & 1--16 \\
$r$ & Eviction retention fraction & 0.25--1 \\
$\sigma$ & Prefix sharing factor & 1--2.91 \\
$\phi$ & Fractional allocation waste & $< 0.04$ paged \\
$\lambda$ & Residency vector over (HBM, DRAM, NVMe) & $(1, 0, 0)$ \\
$F(s,b)$ & Per-step FLOPs, Eq.~\eqref{eq:decode_flops} & --- \\
$B(s,b)$ & Per-step bytes, Eq.~\eqref{eq:decode_bytes} & --- \\
$\mathrm{AI}$ & Arithmetic intensity $F/B$, Eq.~\eqref{eq:ai_general} & 2.6 (W1) \\
$I^{*}$ & Ridge point, compute over bandwidth & 281 FLOP/B \\
$\mu$ & Prefill model FLOP utilization & 0.70 \\
$\Pi$ & Peak dense BF16 throughput & 2250 TFLOP/s \\
$C_{hbm}$ & Device HBM capacity & 192 GB \\
$\beta_{\mathrm{agg}}$ & Aggregate HBM bandwidth & 8.0 TB/s \\
$\beta_{p}$ & Per-partition bandwidth, $\beta_{\mathrm{agg}}/D$ & 4.0 TB/s \\
$\beta_{f}$ & Die-to-die fabric bandwidth & 10 TB/s (NV-HBI) \\
$D$ & Package partitions & 2 \\
$\gamma$ & Local fraction of a KV read & --- \\
$\beta_{\mathrm{eff}}$ & Effective bandwidth, Eq.~\eqref{eq:beta_eff} & 8.0 TB/s \\
$s^{*}$ & Traffic crossover, Eq.~\eqref{eq:crossover_traffic} & 427.2k ($b = 1$) \\
$s_{\mathrm{cap}}$ & Capacity crossover, Eq.~\eqref{eq:crossover_capacity} & 159k \\
$b^{\dagger}$ & Intensity fixed point, Eq.~\eqref{eq:fixed_point} & 8 \\
$R_{\infty}$ & Rate-context ceiling, Eq.~\eqref{eq:tokrate} & $24.4$ M tok/s \\
$T_{\mathrm{step}}$ & Decode step time, $B / \beta_{\mathrm{eff}}$ & 22.7 ms (W1) \\
$\epsilon_{q}, \rho_{\mathrm{ret}}$ & Quality budget and retrieval flag, Algorithm~\ref{alg:domain_selection} & --- \\
\bottomrule
\end{tabular}
\end{table}

\subsection{Units and Numerical Conventions}
\label{app:protocol:units}

All units are decimal SI: k $= 10^{3}$, M $= 10^{6}$, GB $= 10^{9}$ bytes, TB/s $= 10^{12}$ bytes per second. Context lengths are decimal token counts: 128k means $128{,}000$ tokens, not $2^{17}$. Under binary units the 128k Llama-3-70B cache is $39.0$ GiB; the paper reports $41.9$ GB, and every derived quantity is computed in decimal units throughout. Vendor figures (dense tensor throughput, aggregate HBM bandwidth, fabric and link rates) are specifications used at face value, not measurements. Derived values are rounded to three significant figures, so a recomputed product may differ in the last digit from the displayed rounded operands. The symbols $s^{*}$, $s_{\mathrm{cap}}$, and $R_{\infty}$ refer to the BF16 cache ($w_{kv} = 2$ bytes) unless a precision is stated.

\subsection{Evidence Classes and Mixing Rules}
\label{app:protocol:evidence}

Every quantitative claim in the paper carries one of three evidence classes. The tables carry two of them, derived and reported, never mixed within a cell; the third, analytical, is the closed forms of Section~3 from which the derived values are computed.

\begin{enumerate}[leftmargin=1.5em]
  \item \textbf{Analytical.} Closed-form algebra over the cost model: Eqs.~\eqref{eq:decode_flops}--\eqref{eq:speedup} and Propositions~\ref{prop:ai_limits}--\ref{prop:beta_eff}. Exact given the assumptions of Appendix~\ref{app:derivations}.
  \item \textbf{Derived.} Point evaluations of the closed forms at the constants of Tables~\ref{tab:model_kv} and~\ref{tab:hardware}: Table~\ref{tab:crossover}, Table~\ref{tab:worked_decode}, and the W1 and W2 speedup columns of Table~\ref{tab:sok_matrix}. Derived values name their defining equation and assume perfect kernel efficiency.
  \item \textbf{Reported.} Within-paper values quoted from the cited literature with the benchmark named, confined to the quality column of Table~\ref{tab:sok_matrix} and the domain discussions of Section~4. Reported values are never compared across papers, because the underlying workloads differ.
\end{enumerate}

\subsection{Workload Parameterization and Measurement Assumptions}
\label{app:protocol:workloads}

Table~\ref{tab:app_workloads} restates the workloads of Table~\ref{tab:workloads} with full parameterization. All derived numbers in the paper additionally assume: (i) a FlashAttention-class attention kernel that reads the cache exactly once per step, with no recomputation and no score spill to HBM \cite{dao2022flashattention}; (ii) decode step time at the roofline bound $T_{\mathrm{step}} = B(s,b)/\beta_{\mathrm{eff}}$, with striped placement on multi-die packages unless stated; (iii) prefill at model FLOP utilization $\mu = 0.70$ of dense BF16 peak; (iv) dequantization and gather overheads excluded from the speedup columns, with realized overheads discussed in Section~5; and (v) quality deltas imported only as within-paper reported values.

\begin{table}[t]
\centering
\caption{Full workload parameterization. All workloads use Llama-3-70B (Table~\ref{tab:model_kv}) on B200 (Table~\ref{tab:hardware}) at 128k context.}
\label{tab:app_workloads}
\footnotesize
\begin{tabular}{cclll}
\toprule
ID & $b$ & Structure & Metrics & Binding equations \\
\midrule
W1 & 1 & No sharing & $T_{\mathrm{step}}$, tok/s & Eqs.~\eqref{eq:decode_bytes}, \eqref{eq:crossover_traffic} \\
W2 & 32 & No sharing & tok/s, speedup & Eqs.~\eqref{eq:ai_general}, \eqref{eq:speedup} \\
W3 & 8 & 96k shared prefix $+$ 32k unique & TTFT, stored bytes & Eqs.~\eqref{eq:kv_bytes_unified}, \eqref{eq:app_prefill_time} \\
W4 & 1 & Full reuse after 10-minute gap & fetch versus recompute & Eq.~\eqref{eq:app_prefill_time} and link rates of Section~2 \\
\bottomrule
\end{tabular}
\end{table}

\bibliographystyle{IEEEtran}
\bibliography{references}

@inproceedings{vaswani2017attention,
  title={Attention is All You Need},
  author={Vaswani, Ashish and Shazeer, Noam and Parmar, Niki and Uszkoreit, Jakob and Jones, Llion and Gomez, Aidan N. and Kaiser, {\L}ukasz and Polosukhin, Illia},
  booktitle={Advances in Neural Information Processing Systems (NeurIPS)},
  year={2017}
}

@article{llama3ai2024llama3,
  title={The Llama 3 Herd of Models},
  author={{AI@Meta}},
  journal={arXiv preprint arXiv:2407.21783},
  year={2024}
}

@article{williams2009roofline,
  title={Roofline: An Insightful Visual Performance Model for Multicore Architectures},
  author={Williams, Samuel and Waterman, Andrew and Patterson, David},
  journal={Communications of the ACM},
  volume={52},
  number={4},
  pages={65--76},
  year={2009}
}

@misc{nvidia2023h100,
  title={{NVIDIA H100 Tensor Core GPU Datasheet}},
  author={{NVIDIA Corporation}},
  year={2023},
  howpublished={\url{https://resources.nvidia.com/en-us-tensor-core/nvidia-tensor-core-gpu-datasheet}}
}

@misc{nvidia2024b200,
  title={{NVIDIA Blackwell B200 Tensor Core GPU Datasheet}},
  author={{NVIDIA Corporation}},
  year={2024},
  howpublished={\url{https://www.nvidia.com/en-us/data-center/b200/}}
}

@misc{amd2024mi300x,
  title={{AMD Instinct MI300X Accelerator Datasheet}},
  author={{Advanced Micro Devices}},
  year={2024},
  howpublished={\url{https://www.amd.com/en/products/accelerators/instinct/mi300/mi300x.html}}
}

@article{shazeer2019mqa,
  title={Fast Transformer Decoding: One Write-Head is All You Need},
  author={Shazeer, Noam},
  journal={arXiv preprint arXiv:1911.02150},
  year={2019}
}

@inproceedings{ainslie2023gqa,
  title={{GQA}: Training Generalized Multi-Query Transformer Models from Multi-Head Checkpoints},
  author={Ainslie, Joshua and Lee-Thorp, James and de Jong, Michiel and Zemlyanskiy, Yury and Lebr{\'o}n, Federico and Sanghai, Sumit},
  booktitle={Proceedings of EMNLP},
  year={2023}
}

@article{deepseekai2024deepseekv2,
  title={{DeepSeek-V2}: A Strong, Economical, and Efficient Mixture-of-Experts Language Model},
  author={{DeepSeek-AI}},
  journal={arXiv preprint arXiv:2405.04434},
  year={2024}
}

@inproceedings{liu2024kivi,
  title={{KIVI}: A Tuning-Free Asymmetric 2bit Quantization for {KV} Cache},
  author={Liu, Zirui and Yuan, Jiayi and Jin, Hongye and Zhong, Shaochen and Xu, Zhaozhuo and Braverman, Vladimir and Chen, Beidi and Hu, Xia},
  booktitle={Proceedings of ICML},
  year={2024}
}

@inproceedings{hooper2024kvquant,
  title={{KVQuant}: Towards 10 Million Context Length {LLM} Inference with {KV} Cache Quantization},
  author={Hooper, Coleman and Kim, Sehoon and Mohammadzadeh, Hiva and Mahoney, Michael W. and Shao, Yakun Sophia and Keutzer, Kurt and Gholami, Amir},
  booktitle={Advances in Neural Information Processing Systems (NeurIPS)},
  year={2024}
}

@inproceedings{lin2024qserve,
  title={{QServe}: {W4A8KV4} Quantization and System Co-design for Efficient {LLM} Serving},
  author={Lin, Yujun and Tang, Haotian and Yang, Shang and Zhang, Zhekai and Xiao, Guangxuan and Gan, Chuang and Han, Song},
  booktitle={Proceedings of MLSys},
  year={2025}
}

@inproceedings{xiao2023streamingllm,
  title={Efficient Streaming Language Models with Attention Sinks},
  author={Xiao, Guangxuan and Tian, Yuandong and Chen, Beidi and Han, Song and Lewis, Mike},
  booktitle={Proceedings of ICLR},
  year={2024}
}

@inproceedings{zhang2023h2o,
  title={{H2O}: Heavy-Hitter Oracle for Efficient Generative Inference of Large Language Models},
  author={Zhang, Zhenyu and Sheng, Ying and Zhou, Tianyi and Chen, Tianlong and Zheng, Lianmin and Cai, Ruisi and Song, Zhao and Tian, Yuandong and R{\'e}, Christopher and Barrett, Clark and Wang, Zhangyang and Chen, Beidi},
  booktitle={Advances in Neural Information Processing Systems (NeurIPS)},
  year={2023}
}

@inproceedings{li2024snapkv,
  title={{SnapKV}: {LLM} Knows What You are Looking for Before Generation},
  author={Li, Yuhong and Huang, Yingbing and Yang, Bowen and Venkitesh, Bharat and Locatelli, Acyr and Ye, Hanchen and Cai, Tianle and Lewis, Patrick and Chen, Deming},
  booktitle={Advances in Neural Information Processing Systems (NeurIPS)},
  year={2024}
}

@inproceedings{adnan2024keyformer,
  title={Keyformer: {KV} Cache Reduction through Key Tokens Selection for Efficient Generative Inference},
  author={Adnan, Muhammad and Arunkumar, Akhil and Jain, Gaurav and Nair, Prashant and Soloveychik, Ilya and Kamath, Purushotham},
  booktitle={Proceedings of MLSys},
  year={2024}
}

@inproceedings{ge2023fastgen,
  title={Model Tells You What to Discard: Adaptive {KV} Cache Compression for {LLMs}},
  author={Ge, Suyu and Zhang, Yunan and Liu, Liyuan and Zhang, Minjia and Han, Jiawei and Gao, Jianfeng},
  booktitle={Proceedings of ICLR},
  year={2024}
}

@inproceedings{kwon2023vllm,
  title={Efficient Memory Management for Large Language Model Serving with {PagedAttention}},
  author={Kwon, Woosuk and Li, Zhuohan and Zhuang, Siyuan and Sheng, Ying and Zheng, Lianmin and Yu, Cody Hao and Gonzalez, Joseph and Zhang, Hao and Stoica, Ion},
  booktitle={Proceedings of SOSP},
  year={2023}
}

@inproceedings{zheng2024sglang,
  title={{SGLang}: Efficient Execution of Structured Language Model Programs},
  author={Zheng, Lianmin and Yin, Liangsheng and Xie, Zhiqiang and Sun, Chuyue and Huang, Jeff and Yu, Cody Hao and Cao, Shiyi and Kozyrakis, Christos and Stoica, Ion and Gonzalez, Joseph E. and Barrett, Clark and Sheng, Ying},
  booktitle={Advances in Neural Information Processing Systems (NeurIPS)},
  year={2024}
}

@inproceedings{sheng2023flexgen,
  title={{FlexGen}: High-Throughput Generative Inference of Large Language Models with a Single {GPU}},
  author={Sheng, Ying and Zheng, Lianmin and Yuan, Binhang and Li, Zhuohan and Ryabinin, Max and Chen, Beidi and Liang, Percy and R{\'e}, Christopher and Stoica, Ion and Zhang, Ce},
  booktitle={Proceedings of ICML},
  year={2023}
}

@inproceedings{liu2024cachegen,
  title={{CacheGen}: {KV} Cache Compression and Streaming for Fast Large Language Model Serving},
  author={Liu, Yuhan and Li, Hanchen and Cheng, Yihua and Ray, Siddhant and Huang, Yuyang and Zhang, Qizheng and Du, Kuntai and Yao, Jiayi and Lu, Shan and Ananthanarayanan, Ganesh and Maire, Michael and Hoffmann, Henry and Holtzman, Ari and Jiang, Junchen},
  booktitle={Proceedings of ACM SIGCOMM},
  year={2024}
}

@misc{lmcache2024,
  title={{LMCache}: An Efficient {KV} Cache Layer for Enterprise-Scale {LLM} Inference},
  author={{LMCache Team}},
  year={2024},
  howpublished={\url{https://github.com/LMCache/LMCache}}
}

@article{jiang2024mixtral,
  title={Mixtral of Experts},
  author={Jiang, Albert Q. and Sablayrolles, Alexandre and Roux, Antoine and Mensch, Arthur and Savary, Blanche and Bamford, Chris and Chaplot, Devendra Singh and de las Casas, Diego and Hanna, Emma Bou and Bressand, Florian and others},
  journal={arXiv preprint arXiv:2401.04088},
  year={2024}
}

@article{touvron2023llama2,
  title={Llama 2: Open Foundation and Fine-Tuned Chat Models},
  author={Touvron, Hugo and Martin, Louis and Stone, Kevin and others},
  journal={arXiv preprint arXiv:2307.09288},
  year={2023}
}

@inproceedings{dao2022flashattention,
  title={{FlashAttention}: Fast and Memory-Efficient Exact Attention with {IO}-Awareness},
  author={Dao, Tri and Fu, Daniel Y. and Ermon, Stefano and Rudra, Atri and R{\'e}, Christopher},
  booktitle={Advances in Neural Information Processing Systems (NeurIPS)},
  year={2022}
}

@inproceedings{yu2022orca,
  title={Orca: A Distributed Serving System for Transformer-Based Generative Models},
  author={Yu, Gyeong-In and Jeong, Joo Seong and Kim, Geon-Woo and Kim, Soojeong and Chun, Byung-Gon},
  booktitle={Proceedings of OSDI},
  year={2022}
}

@inproceedings{patel2024splitwise,
  title={{Splitwise}: Efficient Generative {LLM} Inference Using Phase Splitting},
  author={Patel, Pratyush and Choukse, Esha and Zhang, Chaojie and Shah, Aashaka and Goiri, {\'I}{\~n}igo and Maleki, Saeed and Bianchini, Ricardo},
  booktitle={Proceedings of ISCA},
  year={2024}
}

@inproceedings{zhong2024distserve,
  title={{DistServe}: Disaggregating Prefill and Decoding for Goodput-Optimized Large Language Model Serving},
  author={Zhong, Yinmin and Liu, Shengyu and Chen, Junda and Hu, Jianbo and Zhu, Yibo and Liu, Xuanzhe and Jin, Xin and Zhang, Hao},
  booktitle={Proceedings of OSDI},
  year={2024}
}

@article{singh2026deployment,
  title={Optimizing {AI} Inference Across the Deployment Stack},
  author={Singh, Tejinder and Pflueger, J. and Mitra, J. and Lincourt, R. and Markow, M. and Patel, B. A.},
  journal={arXiv preprint arXiv:2609.10550},
  year={2026}
}

@article{xu2026kvstrategies,
  title={{KV} Cache Optimization Strategies for Scalable and Efficient {LLM} Inference},
  author={Xu, Y. and Khaira, N. K. and Singh, Tejinder},
  journal={arXiv preprint arXiv:2603.20397},
  year={2026}
}

@article{vonlaszewski2025carpentry,
  title={{AI} Benchmark Democratization and Carpentry},
  author={von Laszewski, Gregor and Brewer, W. and Thiyagalingam, J. and Papay, J. and Foundjem, A. and others},
  journal={arXiv preprint arXiv:2512.11588},
  year={2025}
}

@inproceedings{tang2024quest,
  title={{Quest}: Query-Aware Sparsity for Efficient Long-Context {LLM} Inference},
  author={Tang, Jiaming and Zhao, Yilong and Zhu, Kan and Xiao, Guangxuan and Kasikci, Baris and Han, Song},
  booktitle={Proceedings of ICML},
  year={2024}
}

@inproceedings{jiang2024minference,
  title={{MInference} 1.0: Accelerating Pre-filling for Long-Context {LLMs} via Dynamic Sparse Attention},
  author={Jiang, Huiqiang and Li, Yucheng and Zhang, Chengruidong and Wu, Qianhui and Luo, Xufang and Ahn, Surin and Han, Zhenhua and Abdi, Amir H. and Li, Dongsheng and Lin, Chin-Yew and Yang, Yuqing and Qiu, Lili},
  booktitle={Advances in Neural Information Processing Systems (NeurIPS)},
  year={2024}
}

@article{cai2024pyramidkv,
  title={{PyramidKV}: Dynamic {KV} Cache Compression based on Pyramidal Information Funneling},
  author={Cai, Zefan and Zhang, Yichi and Gao, Bofei and Liu, Yuliang and Li, Yucheng and Liu, Tianyu and Lu, Keming and Xiong, Wayne and Dong, Yue and Hu, Junjie and Xiao, Wen},
  journal={arXiv preprint arXiv:2406.02069},
  year={2024}
}

@inproceedings{yao2025cacheblend,
  title={{CacheBlend}: Fast Large Language Model Serving for {RAG} with Cached Knowledge Fusion},
  author={Yao, Jiayi and Li, Hanchen and Liu, Yuhan and Ray, Siddhant and Cheng, Yihua and Zhang, Qizheng and Du, Kuntai and Lu, Shan and Jiang, Junchen},
  booktitle={Proceedings of EuroSys},
  year={2025}
}

\end{document}